\documentclass[12pt]{article}
\pdfoutput=1

\usepackage[margin=1in]{geometry}
\usepackage{amsthm,amsmath,amsfonts,mathtools}
\usepackage{thmtools}
\usepackage{thm-restate}
\usepackage[pagebackref]{hyperref}
\usepackage{graphicx}
\usepackage{caption}
\usepackage{subcaption}
\usepackage{float}
\restylefloat{table}
\usepackage{siunitx}
\usepackage{bbm}
\usepackage{bm}
\usepackage{comment}
\usepackage{xcolor}
\usepackage[shortlabels]{enumitem}
\usepackage{algorithm}
\usepackage{algorithmic}
\usepackage{makecell}
\usepackage{titletoc}
\usepackage{tikz}
\usepackage{array,colortbl,multirow,multicol,booktabs}
\usepackage{microtype}
\usepackage[nameinlink,capitalize]{cleveref}
\usepackage{wrapfig}
\newcommand{\footremember}[2]{%
    \footnote{#2}
    \newcounter{#1}
    \setcounter{#1}{\value{footnote}}%
}
\newcommand{\footrecall}[1]{%
    \footnotemark[\value{#1}]%
}
\hypersetup{
    pdftitle={Triply efficient shadow tomography}, 
    pdfauthor={Robbie King, David Gosset, Robin Kothari, Ryan Babbush}, 
    colorlinks=true, 
    linkcolor=blue, 
    citecolor=blue, 
    urlcolor=blue 
}

\renewcommand{\backref}[1]{}

\renewcommand{\backrefalt}[4]{%
\ifcase #1 %
\or 
[p.\ #2]%
\else 
[pp.\ #2]%
\fi}

\newtheorem{theorem}{Theorem}
\newtheorem{lemma}[theorem]{Lemma}
\newtheorem{corol}[theorem]{Corollary}

\newtheorem{definition}[theorem]{Definition}
\newtheorem{conjecture}[theorem]{Conjecture}

\newtheorem{claim}[theorem]{Claim}

\DeclareMathOperator{\Tr}{Tr}

\DeclareMathOperator{\poly}{poly}
\DeclareMathOperator{\sign}{sign}

\DeclareMathOperator{\codeg}{codeg}

\renewcommand{\varepsilon}{\epsilon}
\newcommand{\eps}{\epsilon}

\begin{document}
\title{\vspace{-3ex}Triply efficient shadow tomography}

\author{Robbie King\footremember{google}{Google Quantum AI, Venice, CA, USA}\footremember{caltech}{California Institute of Technology, Pasadena, CA, USA}
\and
David Gosset\footrecall{google} \footremember{iqc}{IQC and Department of Combinatorics and Optimization, University of Waterloo, Canada}\footremember{pi}{Perimeter Institute for Theoretical Physics, Waterloo, Canada}
\and
Robin Kothari\footrecall{google}
\and
Ryan Babbush\footrecall{google}
}
\maketitle
\begin{abstract}
Given copies of a quantum state $\rho$, a shadow tomography protocol aims to learn all expectation values from a fixed set of observables, to within a given precision $\epsilon$. We say that a shadow tomography protocol is \textit{triply efficient} if it is sample- and time-efficient, and only employs measurements that entangle a constant number of copies of $\rho$ at a time.  The classical shadows protocol based on random single-copy measurements is triply efficient for the set of local Pauli observables. 
This and other protocols based on random single-copy Clifford measurements can be understood as arising from fractional colorings of a graph $G$ that encodes the commutation structure of the set of observables. Here we describe a framework for two-copy shadow tomography that uses an initial round of Bell measurements to reduce to a fractional coloring problem in an induced subgraph of $G$ with bounded clique number. This coloring problem can be addressed using techniques from graph theory known as \textit{chi-boundedness}.
Using this framework we give the first triply efficient shadow tomography scheme for the set of local fermionic observables, which arise in a broad class of interacting fermionic systems in physics and chemistry. We also give a triply efficient scheme for the set of all $n$-qubit Pauli observables. 
Our protocols for these tasks use two-copy measurements, which is necessary: sample-efficient schemes are provably impossible using only single-copy measurements. Finally, we give a shadow tomography protocol that compresses an $n$-qubit quantum state into a $\poly(n)$-sized classical representation, from which one can extract the expected value of any of the $4^n$ Pauli observables in $\poly(n)$ time, up to a small constant error.
\end{abstract}

\section{Introduction}

In this paper we discuss protocols for shadow tomography, as introduced by Aaronson~\cite{aaronson2018shadow}, specialized to the case of Pauli observables. Let 
\begin{equation}
\mathcal{P}^{(n)}=\{\mathbbm{1},X,Y,Z\}^{\otimes n}    
\end{equation}
be the set of $n$-qubit Pauli operators and consider a subset $S\subseteq \mathcal{P}^{(n)}$. Let $\rho$ be an unknown $n$-qubit quantum state. Given copies of $\rho$, we would like to learn the expectation values $\Tr(P\rho)$ to precision $\varepsilon$ for every $P \in S$.

\begin{definition}[Pauli shadow tomography]
The shadow tomography task for a set $S\subseteq\mathcal{P}^{(n)}$ is as follows. We are given copies of an unknown $n$-qubit state $\rho$, and our goal is to output estimates $y_P$ such that, 
with high probability\footnote{Throughout this paper, we use ``with high probability'' to mean with probability at least $99\%$, say.} we have
$|y_P - \Tr(P\rho)|\leq \eps$ for all $P\in S$.
\end{definition}

One can use a very naive tomography protocol to perform this task. For a given Pauli $P\in S$, if we measure its value $O((\log{|S|})/\eps^2)$ times (using one copy of $\rho$ for each measurement), then we can ensure that the sample mean is within $\eps$ of $\mathrm{Tr}(\rho P)$ with probability at least $1-0.01/|S|$. If we follow this procedure for each of the $|S|$ Paulis in $S$, the union bound guarantees that, with high probability they will all be $\eps$-close to their true values. This algorithm uses $O((|S|\log{|S|})/\eps^{2})$ copies of the unknown state $\rho$ and is computationally efficient. 

Remarkably, Aaronson described a protocol for shadow tomography of any set of bounded observables (such as Pauli observables) that uses exponentially fewer copies of the unknown state than the naive algorithm~\cite{aaronson2018shadow, aaronson2019gentle, brandao2017quantum}. The best known scaling of the number of samples for learning $m$ general observables is $O(n (\log^2{m}) / \varepsilon^4)$ \cite{buadescu2021improved}. However, the general shadow tomography schemes suffer from two major caveats: they are explicitly exponential in computational runtime (even when the number of observables $m$ scales polynomially with $n$), and they require entangled measurements on many copies of the unknown state $\rho$ at a time.

These caveats can be avoided for certain restricted sets of observables such as low-weight Pauli operators.  For $k$-local Paulis with $k=O(1)$, there are simple and computationally efficient protocols  to learn $m$ observables with $O((\log{m}) / \eps^2)$ single-copy measurements \cite{cotler2020quantum,evans2019scalable, bonet2020nearly, jiang2020optimal, huang2020predicting}. The classical shadows framework \cite{huang2020predicting} provides a broader family of learning protocols that can also handle other sets of non-Pauli observables with single-copy measurements, notably including rank-$1$ observables which are relevant to fidelity estimation. However, classical shadows and other single-copy learning strategies become inefficient for higher weight Pauli operators. 

To go beyond low-weight Paulis one can use a shadow tomography protocol developed in \cite{huang2021information} which learns \textit{any} subset $S$ of Pauli operators using  $O((\log{|S|})/\eps^4)$ copies of $\rho$, and $\mathrm{poly}(|S|,n,1/\epsilon)$ runtime. 
The protocol proceeds in two stages. In the first stage---\textit{learning magnitudes}---one computes estimates of the magnitudes $|\mathrm{Tr}(\rho P)|$ to within $\eps/4$ error (say), for all Paulis $P\in S$. Remarkably, this can be achieved efficiently using only $O((\log{|S|})/\eps^{4})$ two-copy measurements using the well-known Bell sampling procedure \cite{montanaro2017learning}. It is based on measuring copies of $\rho\otimes \rho$ in the basis which simultaneously diagonalizes the operators $P\otimes P$ for all $P\in \mathcal{P}^{(n)}$. In the second stage---\textit{learning signs}---one computes the signs of all Paulis $P\in S$ that were estimated to have nonnegligible magnitude in the first stage. The learning signs protocol from Ref.~\cite{huang2021information} proceeds by a sequence of ``gentle measurements" on $O((\log|S|)/\eps^2)$ copies of $\rho$.

The requirement to perform joint entangled measurements on many copies of $\rho$ is a significant drawback. If the state $\rho$ is prepared by a real-world quantum computer, it may be prohibitive to require a quantum memory whose size is many times larger than the quantum system of interest.  It is natural to ask if we can achieve sample- and computational efficiency while using measurements on only a few copies of $\rho$ at a time. We call a shadow tomography protocol with these features \textit{triply efficient}.

\vspace{0.25cm}

\noindent\fbox{\begin{minipage}{\textwidth}
\vspace{0.05cm}
\paragraph{Triply efficient shadow tomography \vspace{0.4cm}}

\begin{enumerate}[]
    \item{\textit{Sample efficiency}: The number of samples scales as $\poly(\log |S|, 1/\eps)$.}
    \item{\textit{Computational efficiency}: The classical and quantum computation is $\poly(|S|,n,1/\eps)$.}
    \item{\textit{Few-copy measurements}: The algorithm uses joint measurements on a constant number of copies of $\rho$ (ideally 1 or 2).\vspace{0.2cm}}
\end{enumerate}
\end{minipage}}

\vspace{0.25cm}

In addition to the above, it would be reasonable to ask that the total quantum memory used by the shadow tomography algorithm is $O(n)$; all of the protocols we propose in this work satisfy this stronger requirement.

Is there a triply efficient shadow tomography protocol? While this question is well-posed for arbitrary subsets of observables $S\subseteq \mathcal{P}^{(n)}$, we restrict our attention to three subsets that are practically motivated and representative of the complexity of the problem:
\begin{enumerate}[]
    \item{$\mathcal{P}^{(n)}_k$: The set of $k$-local Pauli operators on $n$ qubits, where $k=O(1)$.}
    \item{$\mathcal{F}^{(n)}_k$: The set of $k$-body fermionic operators on $n$ fermionic modes, where $k=O(1)$.}
    \item{$\mathcal{P}^{(n)}$: The set of all Pauli operators on $n$ qubits.}
\end{enumerate}

Note that a triply efficient shadow tomography protocol for the set of all Paulis does not directly give one for the set of $k$-local Paulis or $k$-body fermionic operators, even though they are subsets of the set of all Paulis. This is because the complexity requirements are functions of $|S|$. Thus a triply efficient protocol for the set of all Paulis can use more samples and time than we allow for smaller subsets of Paulis.

As we describe below, $k$-local Pauli operators and $k$-body fermionic operators arise in a variety of applications in many-body physics and quantum chemistry and are one of the most common algorithmic applications of shadow tomography. 
The set of all Pauli operators is of exponential size in $n$, and perhaps not as practically relevant, but seems important to study nonetheless since it represents a general Pauli learning task. Surprising tomography algorithms are possible in this case, including an algorithm we describe that, for any constant error $\eps$,  compresses the output (of size $4^n$) into a polynomial-sized description from which an $\eps$-estimate of the expected value of any Pauli observable can be extracted efficiently.

This brings us to the main question that guided this work:
\begin{quotation}
    \noindent\emph{Do there exist triply efficient shadow tomography protocols for these observables?}
\end{quotation}

In fact, the schemes based on random single-copy measurements described in Refs.~\cite{huang2020predicting,cotler2020quantum,bonet2020nearly,jiang2020optimal,evans2019scalable} already achieve triply efficient shadow tomography for the set $\mathcal{P}^{(n)}_k$ of $k$-local Pauli operators with $k=O(1)$. In this paper we  present triply efficient shadow tomography algorithms for the set $\mathcal{F}^{(n)}_k$ of $k$-body fermionic operators for $k=O(1)$, and the set $\mathcal{P}^{(n)}$ of all $n$-qubit Pauli operators. Furthermore, our algorithms only use Clifford measurements on $2$ copies of $\rho$ at a time.

\begin{table}[ht]
\centering
\begin{tabular}{lc}
\toprule
Observables & Triply efficient shadow tomography? \\ \midrule
$k$-local Pauli operators & Yes \cite{huang2020predicting,cotler2020quantum,bonet2020nearly,jiang2020optimal,evans2019scalable} \\
$k$-body fermionic operators & \cellcolor{blue!25}Yes (\Cref{thm:majorana_informal}) \\
All Pauli operators & \cellcolor{blue!25} Yes (\Cref{thm:mmw_informal})\\ \bottomrule
\end{tabular}
\caption{A qualitative summary of triply efficient shadow tomography. Cells in blue are new results in this paper.}
\end{table}

We will see that it is impossible to perform sample-efficient shadow tomography using only single-copy measurements for the set of $k$-body fermionic operators (\Cref{thm:majorana_lower_bd_informal}), as well as for the set of all Paulis \cite{chen2022exponential}. Taken together, our protocols and the single-copy lower bounds demonstrate that two-copy measurements are necessary and sufficient for Pauli and fermionic shadow tomography. Further, we see a striking difference between local Paulis, for which single-copy measurements suffice, and local fermionic operators where entangled measurements are necessary.

\subsection{Local observables}
\label{sec:local}
In order to describe our results, let us now define the sets of local observables that are relevant to qubit and fermionic systems. Let $|P|$ denote the Pauli-weight of an operator $P\in \mathcal{P}^{(n)}$, i.e., the number of qubits on which it acts nontrivially. For example, $|X \otimes \mathbbm{1} \otimes Y \otimes \mathbbm{1} \otimes Z| = 3$.

A broad class of quantum many-body systems that arise in condensed matter physics are described by systems of spins with $k=O(1)$ particle interactions. Such systems are described by a Hamiltonian operator which can be expressed as a sum of operators from the set
\begin{equation}
\mathcal{P}^{(n)}_k=\{P\in \mathcal{P}^{(n)}: |P|=k\} \qquad \qquad \textbf{($k$-local Pauli operators)}
\label{eq:qubitlocality}
\end{equation}
of all weight-$k$ Pauli observables. Note that $|\mathcal{P}^{(n)}_k|=3^k \binom{n}{k}$ and $\log|\mathcal{P}^{(n)}_k|=O(k \log n)$. Learning all Paulis in the set $\mathcal{P}^{(n)}_k$ is quite useful---it allows one to reconstruct all the $k$-qubit reduced density matrices of the state $\rho$ and compute for example the expected value of any $k$-local Hamiltonian operator. 
So shadow tomography with the set $\mathcal{P}^{(n)}_k$ is particularly relevant for characterizing ground states of quantum spin systems with few-body interactions.

A different subset of Pauli operators describes few-body interactions between fermionic particles, such as the electronic structure of molecules. Just as before, there is a fermionic locality parameter $k$ but it is fundamentally different from the one defined by Pauli weight.  To describe it, one fixes any subset of $2n$ anticommuting $n$-qubit Pauli operators: 
\begin{equation}
c_1,c_2,\ldots, c_{2n} \in \mathcal{P}^{(n)} \qquad \forall a,b: c_ac_b+c_bc_a=2\delta_{ab} \mathbbm{1},
\label{eq:majoranas}
\end{equation}
where $\delta_{ab}$ is $1$ if $a=b$ and $0$ otherwise.
Note that since $c_a\in \mathcal{P}^{(n)}$ we also have $c_a^{\dagger}=c_a$ for each $1\leq i\leq 2n$. These are known as the Majorana fermion operators associated with a fermionic system with $n$ modes (a fermionic mode is a state that can either be occupied or unoccupied by a fermionic particle). Note that there is a freedom here---a particular choice of operators in \cref{eq:majoranas} is a ``fermion-to-qubit mapping" that describes how we associate the degrees of freedom of the $n$-mode fermionic system with those of the $n$-qubit Hilbert space. Such mappings have a long history and there are several choices that are used in practice to design algorithms for fermionic systems on a quantum computer (see, e.g., Refs~\cite{jordan1928paulische,bravyi2002fermionic, seeley2012bravyi, jiang2020optimal, derby2021compact}, and \Cref{sec:majoranas}). However, our discussion and the results described below apply to any choice of fermion-to-qubit mapping.

With our Majoranas (\cref{eq:majoranas}) in hand, let us now define the Majorana monomials as
\begin{equation}
\Gamma(x)= i^{|x|\cdot(|x|-1)/2} c_1^{x_1}c_2^{x_2}\ldots c_{2n}^{x_{2n}} \qquad \qquad \forall x\in \{0,1\}^{2n}.
\end{equation}
The overall phase factor $i^{|x|\cdot(|x|-1)/2}$ ensures that $\Gamma(x)$ is Hermitian for all $x\in \{0,1\}^{2n}$. In fact, the $4^n$ operators
\begin{equation}
\{\Gamma(x): x\in \{0,1\}^{2n}\}
\label{eq:monomial}
\end{equation}
coincide with the $4^n$ $n$-qubit Pauli operators in $\mathcal{P}^{(n)}$, up to (efficiently computable) signs. Now let us define the $k$-body fermionic operators
\begin{equation}
\mathcal{F}^{(n)}_k=\{\Gamma(x): |x|=2k\},  \qquad \qquad  \textbf{($k$-body fermionic operators)}
\end{equation}
which should be compared with \cref{eq:qubitlocality}. Note that $|\mathcal{F}^{(n)}_k|=\binom{2n}{2k}$ and  $\log |\mathcal{F}^{(n)}_k| = O(k \log n)$.
A system of $n$ fermionic modes with $k$-particle interactions is described by a Hamiltonian operator that is a sum of terms from $\mathcal{F}^{(n)}_k$. Note that $\mathcal{F}^{(n)}_k$ consists of the Majorana monomials in \cref{eq:monomial} of degree $2k$; typically only these even-degree monomials are relevant to physics and chemistry due to conservation of fermionic parity.

The notion of $k$-locality for fermions is fundamentally more expressive that of $k$-locality for spin systems; the former subsumes the latter in the sense that compact ``qubit-to-fermion mappings" exist which embed the $k$-local $n$-qubit Pauli operators within a subset of the $k$-body fermionic operators on $O(n)$ fermionic modes (see for example Ref.~\cite{bravyi2019approximation}), whereas fermion-to-qubit mappings necessarily represent the Majorana fermion operators \cref{eq:majoranas} using Pauli operators of average weight at least $\Omega(\log(n))$ \cite{jiang2020optimal}.

The expectation values of the fermionic operators $\mathcal{F}^{(n)}_k$ comprise the matrix elements of what physicists and chemists refer to as the $k$-body reduced density matrix, or $k$-RDM.
One can efficiently compute most physically relevant local observables of a fermionic system (e.g., dipole moment, charge density, and importantly---energy) from just the 1- and 2-RDMs as a consequence of fermions being identical particles that interact pairwise. A large body of literature exists on methods for computing the fermionic 2-RDM matrix elements as a means of estimating the energy of chemical systems during the course of a quantum variational algorithm (see e.g., \cite{IzmaylovClique,ChanRDMTomo,bonet2020nearly,ZhaoShadow}).

A number of important methods for post-processing the output of quantum simulations actually \emph{require} measuring $k$-RDMs. For example, subspace expansion techniques for approximating excited states from ground states via linear response typically require the 4-RDM \cite{mcclean2017,GeneralizedSubspace}. 
Perturbation theory \cite{ChanChromium,sharma2017} and multi-reference configuration interaction methods \cite{Takeshita2020} for relaxing ground state calculations in small basis sets towards their continuum (large basis) limit often require the 4-RDM but converge even faster given access to higher order RDMs. Finally, there are popular impurity model schemes for extrapolating finite simulations of condensed phase fermionic systems towards their thermodynamic limits (e.g., density matrix embedding theory \cite{Wouters2016}) and hybrid quantum-classical schemes for quantum Monte Carlo \cite{Huggins2022}, that require the full 1-RDM.

\subsection{Results}

\paragraph{Single-copy measurements.}
It would be very practical if we could achieve triply efficient Pauli shadow tomography using only measurements on one copy of $\rho$ at a time---and for $k$-local Paulis, it can be done. 
The shadow tomography task for $\mathcal{P}^{(n)}_k$ can be performed, using a time-efficient algorithm, using only \textit{single-copy} measurements on 
$O(3^k (\log |\mathcal{P}^{(n)}_k|)/\eps^2) = O(3^k (k \log n)/\eps^2)$ copies of $\rho$ \cite{huang2020predicting,cotler2020quantum,bonet2020nearly,jiang2020optimal,evans2019scalable}.

One might hope that we can similarly achieve triply efficient Pauli shadow tomography for $\mathcal{F}^{(n)}_k$ and $\mathcal{P}^{(n)}$ as well. Unfortunately, this is not possible in either case, even if we only care about sample efficiency and allow unbounded computation time.
It was shown in Ref \cite{chen2022exponential} that sample-efficient shadow tomography with single-copy measurements is impossible for the set of all Paulis.

\begin{theorem}[\cite{chen2022exponential}]
There is no sample-efficient shadow tomography protocol with single-copy measurements for the set $S=\mathcal{P}^{(n)}$ of all Paulis. In particular, any protocol based on single-copy measurements must consume $\Omega(2^n / \varepsilon^2)$ copies of $\rho$.
\label{thm:notallpaulis}
\end{theorem}

For $k$-body fermionic operators, one can also establish a lower bound, see \Cref{sec:lower_bound}.\footnote{Theorem 1 of \cite{bonet2020nearly} contains a lower bound for single-copy non-adaptive Clifford measurements; on the other hand, \Cref{thm:majorana_lower_bd_informal} applies to arbitrary and even adaptive single-copy measurements.}

\begin{theorem} \label{thm:majorana_lower_bd_informal}
There is no sample-efficient single-copy shadow tomography protocol for the set $\mathcal{F}^{(n)}_k$ of $k$-body fermionic operators. In particular, for $k=O(1)$, any protocol based on single-copy measurements must consume $\Omega(n^k / \varepsilon^2)$ copies of $\rho$.
\label{thm:notfermions}
\end{theorem}

Several efficient shadow tomography algorithms based on single-copy measurements are known which achieve the $\Omega(n^k/\eps^2)$ lower bound, up to a log factor~\cite{bonet2020nearly, jiang2020optimal, wan2022matchgate, huggins2022nearly}.

The lower bound in \Cref{thm:majorana_lower_bd_informal} demonstrates a significant difference between learning local fermionic observables and local Pauli observables; for local Paulis there are sample-efficient single-copy protocols, whereas this is impossible for local fermionic observables.

In order to achieve triply efficient shadow tomography for $\mathcal{F}^{(n)}_k$ and $\mathcal{P}^{(n)}$, we will need to measure 2 or more copies of $\rho$ at a time. Before jumping into 2-copy measurements, we review the 1-copy algorithm for $k$-local Pauli operators and offer a new interpretation in terms of fractional graph colorings that will be used in our algorithms.

The classical shadows single-copy protocol for $k$-local Paulis is very simple: it is based on  measuring each qubit of $\rho$ (in each copy of $\rho$) in a random single-qubit Pauli basis $X,Y$ or $Z$ uniformly at random \cite{huang2020predicting}. 
The postprocessing of the measurement data to compute expected values is equally simple. 

The high-level format of this protocol is as follows: one selects a Clifford basis at random according to some probability distribution $p$, and then measures in that basis. The distribution has the property that each of the Paulis $P$ in the set $S$ of observables of interest (in the above, $S=\mathcal{P}^{(n)}_k$) has a high chance (at least $3^{-k})$ of being diagonal in a basis sampled from $p$. 
Every time we pick a basis in which a Pauli is diagonal, we learn some information about its expected value and hence the sample complexity of the protocol is inversely related to the  probability of being diagonal in a randomly sampled basis. 

We reinterpret this measurement strategy, and other protocols based on random single-copy Clifford measurements, as arising from fractional colorings of the commutation graph  of the observables $S$, defined as follows:
\begin{definition} \label{def:commutation_graph}
The commutation graph $G(S)$ of a set $S \subseteq \mathcal{P}^{(n)}$ of Pauli operators is the graph with vertex set $S$ and an edge between every pair of anticommuting operators.
\end{definition}

An independent set in $G(S)$ corresponds to a set of commuting observables that can be measured simultaneously via a Clifford measurement. Similarly, a coloring of this graph with $\chi$ colors describes a learning strategy with deterministic single-copy Clifford measurements, based on measuring $\chi$ disjoint sets of commuting Pauli observables. Such deterministic graph coloring strategies for learning Pauli observables have been explored previously, see for example Ref.~\cite{jena2019pauli, IzmaylovClique}.  But the protocol for local Pauli observables described above is not based on a coloring of the commutation graph $G(S)$: the measurement bases are associated with \textit{overlapping} sets of commuting Pauli observables, and are chosen randomly rather than deterministically. As we will see, a probabilistic Clifford measurement strategy can be viewed as defining a \textit{fractional coloring} of $G(S)$. A fractional coloring is a well-studied relaxation of the notion of graph coloring, see \Cref{sec:fractional_colorings} for details. We show that the sample complexity of single-copy learning with Clifford measurements is upper bounded by the fractional chromatic number of $G(S)$, which is the size of the smallest fractional coloring of $G(S)$. In \Cref{sec:fractional_colorings}, we prove the following theorem.

\begin{restatable}{theorem}{fractional}
Let $S\subseteq \mathcal{P}^{(n)}$. Suppose the commutation graph $G(S)$ admits a fractional coloring of size $\chi$ that can be sampled by a classical randomized algorithm with runtime $T$. Then there is an algorithm using only single-copy Clifford measurements of $\rho$ which can estimate $\Tr(P \rho)$ within error $\eps$ for all $P \in S$ with high probability using
\begin{equation}
O(\chi (\log{|S|})/\eps^2)
\label{eq:numfrac}
\end{equation}
copies of $\rho$. The runtime of the algorithm is $O((T+n^3)\cdot \chi (\log{|S|})/\epsilon^{2}) $.
\label{thm:fractional}
\end{restatable}

The single-copy measurement strategies based on fractional colorings from \Cref{thm:fractional} have the special feature that they only use Clifford measurements, which have efficient classical descriptions. Such protocols learn a \textit{compressed classical representation} of $\rho$ that can be used to compute Pauli observables, see \Cref{sec:fractional_colorings} for details.

\paragraph{The power of two copies.}
In light of Theorems \ref{thm:notallpaulis} and \ref{thm:notfermions} we see that there exist sets of Pauli observables for which sample-efficient shadow tomography cannot be achieved with one-copy measurements.  Are two-copy measurements enough? 

Our starting point here is a new algorithm that shows that sample-efficient shadow tomography is indeed possible in the general case with two-copy measurements. The protocol has three steps.

The first step is the learning magnitudes subroutine from Ref.~\cite{huang2021information}, which we now review. In this step we perform Bell sampling to measure copies of $\rho\otimes \rho$ in the Clifford basis that diagonalizes the commuting Pauli observables $P\otimes P$ for all $P\in \mathcal{P}^{(n)}$.  Since these operators commute, we can learn the all observables $\Tr((P \otimes P)(\rho \otimes \rho))=\Tr(\rho P)^2$ for $P \in S$ to error $\delta$ using $O((\log |S|)/\delta^2)$ measurements. By choosing $\delta=\Theta(\eps^2)$, we see that $O((\log|S|)/\eps^{4})$ two-copy measurements suffices to compute estimates $\{u_P\}_{P\in S}$ such that 
\begin{equation}
|u_P-|\mathrm{Tr}(\rho P)||\leq \eps/4 \qquad \text{ for all } \quad P\in S,
\label{eq:epserr}
\end{equation}
with high probability. After we have learned the magnitudes in this way, we may find that some of our estimates are negligible; if our estimate $u_P$ from the first stage is less than $3\eps/4$ then $0$ is an $\eps$-approximation to the expected value $\mathrm{Tr}(\rho P)$ and we can forget about this Pauli $P$ going forward. So in the second stage of the algorithm we are only concerned with observables in the set
\begin{equation}\label{eq:Seps}
S_\eps = \{P\in S: |u_P|\geq 3\eps/4\}.
\end{equation} 
This set $S_\eps$ is a random variable determined by the output of Bell sampling, but the condition in \cref{eq:epserr} implies that with high probability we have
\begin{equation}\label{eq:seps}
    |\Tr(\rho P)| \geq \eps/2 \quad \text{for all} \quad P\in S_{\eps}.
\end{equation}
To complete the learning task it suffices to then compute the sign of $\mathrm{Tr}(\rho P)$ for all Paulis $P\in S_\eps$. To this end, in the second step of the protocol we compute a classical description of a \textit{mimicking state} $\sigma$ such that
\begin{equation}
|\mathrm{Tr}(\sigma P)|\geq \eps/4 \quad \text{for all} \quad P\in S_{\eps}.
\label{eq:mimicking}
\end{equation}
A key observation is that, if  \cref{eq:seps} holds (which occurs with high probability), then (A) there always exists at least one mimicking state $\sigma$ (for instance $\sigma=\rho$ is a valid mimicking state), and (B) given the set $S_{\eps}$, a mimicking state can be found without using any additional copies of $\rho$, by brute-force search. So this second step of the protocol, while computationally inefficient, can be performed without using any additional copies of $\rho$. In the final, third step of the protocol, we now perform Bell sampling on the tensor product $\rho\otimes \sigma$. This will require us to repeatedly prepare the mimicking state $\sigma$ on our quantum computer.

The resulting samples are used to estimate mean values
\begin{equation}
\mathrm{Tr}(P\otimes P (\rho\otimes \sigma))=\mathrm{Tr}(P\rho)\mathrm{Tr}(P\sigma) \quad P\in S_{\eps}.
\label{eq:meansigrho}
\end{equation}
Since $\sigma$ is a mimicking state and satisfies \cref{eq:mimicking}, each of these mean values has magnitude at least $\eps^2/8$ and we can compute all of them up to $\eps^2/16$ error using only $O((\log |S|)/\eps^{4})$ Bell samples. Since $\sigma$ is known to us, we can compute $\mathrm{Tr}(\sigma P)$ exactly (including the sign), and using this knowledge and our estimates of the mean values in \cref{eq:meansigrho} we can infer the sign of each mean value $\mathrm{Tr}(\rho P)$ for $P\in S_{\eps}$. 

In summary, we have described a sample-efficient shadow tomography protocol for any subset of $n$-qubit Pauli observables, that only uses two-copy Clifford measurements.

\begin{theorem}\label{thm:2copy}
There exists a shadow tomography protocol for any subset $S\subseteq \mathcal{P}^{(n)}$ of Pauli observables, that uses only $O((\log{|S|})/\eps^4)$ two-copy Clifford measurements.
\end{theorem}
Note that this matches the sample complexity of the protocol from Ref.~\cite{huang2021information}, while only using two-copy measurements. Although the above protocol is sample efficient, it is extremely inefficient in terms of runtime as it seems to require a brute-force search over the set of all $n$-qubit quantum states. Here we show that the outrageous runtime can be tamed -- the protocol can be modified so that its runtime is upper bounded as $\mathrm{poly}(2^n)$. In this way we obtain a sample- and time-efficient two-copy protocol for the set of all Paulis.

\begin{theorem} \label{thm:mmw_informal}
There exists a triply efficient shadow tomography protocol for the set $S=\mathcal{P}^{(n)}$ of all $n$-qubit Paulis that uses only two-copy Clifford measurements.
In particular, it has sample complexity $O(n \log(n/\eps)/\eps^4)$ and time complexity $\poly(2^n,1/\eps)$.
\end{theorem}

In \Cref{sec:signs_from_magnitudes} we complete the proof of \Cref{thm:mmw_informal} by showing that a suitable mimicking state $\sigma$ satisfying \cref{eq:mimicking} can be computed using $O(\mathrm{poly}(2^n))$ time and $O(n\log(n/\epsilon)/\eps^4 )$ additional single-copy measurements of $\rho$. The algorithm uses the matrix multiplicative weights technique \cite{arora2007combinatorial, aaronson2018online}. Once a classical description of a mimicking state has been computed, we can create the corresponding quantum state using $O(4^n)$ gates~\cite{SBM06}.

Our idea to use a mimicking state computed via the matrix multiplicative weights algorithm has already found an application: Ref. \cite{king2024exponential} applies this to perform shadow tomography on bosonic displacement operators.

\paragraph{Improved algorithms using graph theory.}

The remainder of our results employ a general framework for two-copy shadow tomography which is based on fractional graph coloring. Our framework can be viewed as an extension of a heuristic learning algorithm proposed in Appendix E.2.d of Ref.~\cite{huang2021information}; in this work we use it to obtain algorithms with rigorous performance guarantees.

 We have already seen that single-copy tomography for any set of observables $S$ reduces to fractional graph coloring for the commutation graph $G(S)$. Likewise, via Bell sampling, two-copy tomography reduces to fractional graph coloring for the commutation graph $G(S_{\eps})$. That is, we propose to use the single-copy algorithm to learn all observables in $S_\eps$, once we have already determined $S_\eps$ using an initial stage of Bell sampling. But how do the two-copy measurements help us?

A key insight is that the Paulis in $S_\eps$ cannot be very anticommuting. Intuition from the Heisenberg uncertainty principle tells us that anticommuting (traceless) observables cannot simultaneously be large on a quantum state, since the quantum state cannot simultaneously be an eigenvector of anticommuting observables. But the high probability event in \cref{eq:seps} implies that the Paulis in $S_\eps$ are simultaneously large on the state $\rho$, and thus cannot anticommute with each other too often. This can be formalized as an upper bound on the clique number of their commutation graph. In \Cref{sec:two_copy_learning} we show the following:
\begin{restatable}{lemma}{largestclique}\label{lem:clique_bound_informal}
The largest clique in the commutation graph $G(S_\eps)$ has size at most $4/\eps^{2}$ with high probability.
\end{restatable}

Going forward, our aim is to exploit this upper bound on the clique number to compute good (fractional) colorings. 

Unfortunately, it is well known that the chromatic (or fractional chromatic) number is not upper bounded by any function of the clique number in general \footnote{For example, there exists a family of triangle-free graphs with chromatic number $\Omega(\sqrt{m/\log m})$ where $m$ is the number of vertices~\cite{kim1995ramsey}}.  However, such upper bounds can be established for certain families of graphs, a research direction pioneered by Gy\'arf\'as \cite{gyarfas1987problems}. A family of graphs for which this is possible is called \textit{chi-bounded} and the upper bound on chromatic number is said to be expressed in terms of a chi-binding function (see Refs.~\cite{schiermeyer2019polynomial, scott2020survey} for recent surveys). Our shadow tomography learning task for a set of observables $S$ thus reduces to establishing a suitable chi-binding function for the family of induced subgraphs of the commutation graph $G(S)$, see \Cref{sec:two_copy_learning} for details.

We show that the family of induced subgraphs of the commutation graph of $k$-body fermionic observables admits a polynomial chi-binding function (that does not depend on $n$).
\begin{restatable}{lemma}{chibindfermion} \label{lem:chibindfermion_intro}
Let $k\geq 1$, and let $G'$ be any induced subgraph of the commutation graph $G(\mathcal{F}^{(n)}_k)$ of $k$-body fermionic observables, and let $\omega$ be the size of the largest clique in $G'$. Then the fractional chromatic number of $G'$ satisfies
\begin{equation}
\chi_f(G')\leq p_k(\omega).
\label{eq:chifermi_}
\end{equation}
where $p_k$ is a polynomial. Moreover, for any $k=O(1)$ we can sample from a fractional coloring of $G'$ with size $p_k(\omega)$ using a classical algorithm with runtime $\mathrm{poly}(n)$. The polynomials for $k=1,2$ are $p_1(\omega)=\omega+1$ and $p_2(\omega)=O(\omega^8)$.
\end{restatable}
The proof of \Cref{lem:chibindfermion_intro} is provided in \Cref{sec:kbody}.
As discussed above, the commutation graph $G(S_{\eps})$ is an induced subgraph of $G(S)$ with clique number at most $O(1/\eps^{2})$. For $k$-body fermionic observables $S=\mathcal{F}^{(n)}_k$, \Cref{lem:chibindfermion_intro} tells us there is an efficiently computable fractional coloring of $G(S_{\eps})$ with at most $\mathrm{poly}(1/\eps^{2})$ colors. We can then use the single-copy learning protocol from \Cref{thm:fractional} to learn all observables in $S_\eps$. This reduction, which describes how to convert \Cref{lem:chibindfermion_intro} into a two-copy learning protocol, is formalized in \Cref{lem:twocopycolor}. Putting it all together gives the following theorem.

\begin{theorem} \label{thm:majorana_informal}
Let $k=O(1)$. There exists a triply efficient shadow tomography protocol for the set $S=\mathcal{F}^{(n)}_k$ of $k$-body fermionic observables that uses only two-copy Clifford measurements. 
\end{theorem}

This triply efficient protocol has sample complexity 
\begin{equation}
O\biggl(\frac{\log |\mathcal{F}^{(n)}_k|}{\eps^4} + \frac{p_k(4/\eps^2)\log |\mathcal{F}^{(n)}_k|}{\eps^2} \biggr) =  O((k \log n) p_k(4/\epsilon^2)/\eps^2),    
\end{equation}
where $p_k$ is the polynomial from \Cref{lem:chibindfermion_intro} that depends on the locality $k$, and we also used the fact that $p_k(\omega)=\Omega(\omega)$ for all $k\geq 1$.  For each $k\geq 1$ we obtain an exponential improvement over single-copy learning protocols  in terms of the sample complexity as a function of system size $n$. For $k=1$ our learning algorithm has sample complexity $O((\log{n})/\eps^4)$, and the measurements and postprocesssing are simple to implement. We anticipate that this learning algorithm could find applications in quantum simulations of chemistry and fermionic physics. With our current analysis, the degree of the polynomial $p_k$ increases very rapidly as a function of $k$ rendering the scheme less practical for $k\geq 2$. We hope this could be improved in future work. An upper bound on the $\eps$-dependence of the sample complexity is $\sim \eps^{-O((2k)^{k+1})}$. For $k=2,3$ the sample complexity is $\sim \eps^{-18}$ and $\sim \eps^{-110}$ respectively.

It is natural to ask how far we can push this two-copy framework based on Bell sampling and fractional coloring.  Below, we show that it provides a nontrivial shadow tomography protocol for \textit{any} subset of Pauli observables $S\subseteq \mathcal{P}^{(n)}$. This gives hope that our framework could lead to triply efficient shadow tomography in the general case.

We shall exploit the fact that the longest induced path in the commutation graph $G(\mathcal{P}^{(n)})$ contains at most $2n+1$ vertices, see \Cref{sec:coloring} for details. The following upper bound on chromatic number then follows from a seminal result in chi-boundedness due to Gy\'arf\'as \cite{gyarfas1987problems}; see \Cref{sec:coloring}.
\begin{restatable}{lemma}{chibind}\label{lem:coloring_intro}
Let $G'$ be any induced subgraph of the commutation graph $G(\mathcal{P}^{(n)})$, and let $\omega$ be the size of the largest clique in $G'$. The chromatic number of $G'$ is upper bounded as
\begin{equation}
\chi(G')\leq (2n+1)^{\omega-1}.
\label{eq:chign}
\end{equation}
Moreover, a coloring with this many colors can be computed by a classical algorithm with runtime $\poly(|G'|, n^{\omega})$.
\end{restatable}

To get a shadow tomography algorithm for any set of Pauli observables $S\subseteq P^{(n)}$, we follow the strategy outlined above and formalized in \Cref{lem:twocopycolor}. That is,  we apply \Cref{lem:coloring_intro} to the subgraph $G'=G(S_{\epsilon})$ induced by the set $S_{\epsilon}$ computed via Bell sampling. From \Cref{lem:clique_bound_informal} we have that with high probability the largest clique in $G'$ has size $\omega =O(1/\epsilon^{2})$. So we get an  coloring of $G(S_\eps)$ with at most ${n^{O(1/\eps^2)}}$ colors, that can be computed with runtime $\mathrm{poly}(|S|, n^{1/\epsilon^{2}})$. When $\epsilon=\Omega(1)$ is a small constant, this protocol is time-efficient, has sample complexity $\mathrm{poly}(n)$, and only uses two-copy measurements, for any subset of Pauli observables $S$. This gives a sample-efficient protocol only when $|S|$ is exponentially large as a function of $n$. At a technical level this is a consequence of the factor of $n$ appearing in \cref{eq:chign}, and we do not know if this can be avoided.

This leaves open the question of triply efficient shadow tomography for arbitrary subsets of Pauli observables. 
However, we will see that it provides insight into a related question concerning compressed classical representations of quantum states.

\paragraph{Rapid-retrieval Pauli compression.}

Can we compress an $n$-qubit quantum state into a small amount of classical information, so that the compressed classical description is sufficient to extract the expectation values of any bounded observable to within a small constant error? This question has been studied using tools from communication complexity and it is known that an exponential classical description size is necessary if one wishes to recover  bounded observables in the general case; a representation size $\tilde{\Theta}(\sqrt{2^n})$ is necessary and sufficient for $n$-qubit pure states~\cite{raz1999exponential,gavinsky2007exponential,gosset2018compressed}.

On the other hand, if we restrict our attention to the set of $n$-qubit Pauli observables (or other sets of observables with only singly exponential size), a classical description of size $\mathrm{poly}(n)$ exists and can be computed using the matrix multiplicative weights algorithm  \cite{aaronson2004limitations, aaronson2018online}. However, a significant drawback of known methods for this task is that they require exponential classical runtime to extract the expected value of a given Pauli observable from the compressed classical representation.

A consequence of \Cref{lem:coloring_intro} is that this exponential cost can be avoided, at least for any small constant precision $\eps=\Omega(1)$. That is, one can compress an $n$-qubit state $\rho$ into $\mathrm{poly}(n)$ classical bits. Given this classical data and an $n$-qubit Pauli $P$, there is an \textit{efficient} classical algorithm to estimate $\mathrm{Tr}(\rho P)$ to within $\eps$-error. Moreover, such a representation can be learned from $\mathrm{poly}(n)$ samples of $\rho$. 

\begin{corol}[Rapid-retrieval Pauli compression] \label{cor:compression_intro}
Let $\rho$ be an $n$-qubit quantum state. Let $\eps\in (0,1)$ be a constant independent of $n$. Using two-copy Clifford measurements on $\mathrm{poly}(n)$ copies of $\rho$, along with $2^{O(n)}$ runtime, we can (with high probability) learn a compressed classical representation of $\rho$, call it $D(\rho, \eps)$, that consists of $\mathrm{poly}(n)$ bits. An $\eps$-approximation to the expected value $\mathrm{Tr}(\rho P)$ of any Pauli observable $P\in \mathcal{P}^{(n)}$ can be extracted from $D(\rho, \eps)$ using a classical algorithm with $\mathrm{poly}(n)$ runtime.
\end{corol}

The classical description $D(\rho, \epsilon)$ consists of a list of all the Clifford measurement bases and measurement outcomes used in the two-copy learning algorithm discussed above, see \Cref{sec:commutation} for details. In particular, \Cref{cor:compression_intro} is obtained by combining \Cref{lem:coloring_intro} and \Cref{lem:rapidretrieval}.

\subsection{Discussion and open questions}
In this paper we have provided the first triply efficient shadow tomography protocols for the set of $k$-body fermionic observables and the set of all Pauli operators. We have also provided a route to strengthening and generalizing our results via a connection between two-copy tomography and graph theory techniques related to chi-boundedness.  In \Cref{table} we provide a comparison of our results with other known protocols for shadow tomography.

\begin{table}[htbp]
\small
\renewcommand{\arraystretch}{1.5}
\begin{tabular}{ccccc}
\toprule
Observables & \makecell{Copies per \\ measurement} & \makecell{Algorithm or \\ lower bound}& \makecell{Sample \\ complexity}  & \makecell{Time \\ efficient?} \\ \midrule

\multirow{7}{*}{\makecell{$k$-local \\ Pauli \\ operators}} & \multirow{4}{*}{1} & Naive & $O(3^k \binom{n}{k} k(\log{n})/\eps^2)$ & \checkmark\\
& & \makecell{Various methods \\ \cite{cotler2020quantum,evans2019scalable, bonet2020nearly} \\ \cite{jiang2020optimal, huang2020predicting}} & $O(3^k k (\log n)/\eps^2)$ & \checkmark \\
& & \makecell{Lower bound (\Cref{thm:local_pauli_lowerbd}) \\ for $k \leq \log_3 (2n+1)$.} & \cellcolor{blue!25} $\Omega(3^k/\eps^2)$ & --- \\ \cmidrule{2-5}
& 2 & \Cref{thm:2copy} & \cellcolor{blue!25} $O(k(\log n)/\eps^4)$ & \text{\sffamily X} \\ \cmidrule{2-5}
& unrestricted & \makecell{Bell sampling and gentle\\ measurements \cite{huang2021information}}& $O(k(\log n)/\eps^4)$ & \checkmark\\ 
\midrule

\multirow{10}{*}{\makecell{$k$-body \\ fermionic \\ operators}} & \multirow{3}{*}{1} & Naive & $O\bigl(\binom{2n}{2k} k(\log{n})/\eps^2\bigr)$ & \checkmark\\
& & \makecell{Various methods \\ \cite{bonet2020nearly,jiang2020optimal,wan2022matchgate}} & $O_k(n^k (\log n)/\eps^2)$* & \checkmark \\
& & \makecell{Lower bound (\Cref{thm:local_majorana_lowerbd})} & \cellcolor{blue!25} $\Omega_k(n^k/\eps^2)$*  & --- \\
\cmidrule{2-5}
& \multirow{3}{*}{2} & \Cref{thm:2copy} & \cellcolor{blue!25} $O(k (\log n)/\eps^4)$ & \text{\sffamily X} \\
& & \Cref{thm:majorana_informal} for $k=1$ & \cellcolor{blue!25} $O((\log n)/\eps^4)$ & \checkmark\\
& & \Cref{thm:majorana_informal} & \cellcolor{blue!25} $\log{n} \cdot \poly_k(1/\eps)$$^\dag$ & \checkmark \\
\cmidrule{2-5}
& unrestricted & \makecell{Bell sampling and gentle\\ measurements \cite{huang2021information}}& $O(k (\log n)/\eps^4)$ & \checkmark \\
\midrule

\multirow{6}{*}{\makecell{All Pauli \\ operators}} & \multirow{3}{*}{1} & \makecell{Naive } & $O(4^n n/\eps^2)$ & \checkmark\\
& & \makecell{Random Clifford$^{\ddagger}$} & $O(2^n n/\eps^2)$ & \checkmark\\
& & Lower bound \cite{chen2022exponential} & $\Omega(2^n/\eps^2)$ & --- \\ \cmidrule{2-5}
& \multirow{2}{*}{2} & \Cref{thm:mmw_informal} & \cellcolor{blue!25} $O(n\log(n/\epsilon)/\eps^4)$ & \checkmark\\
&  & \Cref{cor:compression_intro} & \cellcolor{blue!25} $n^{O(1/\eps^2)}$ & \checkmark\\ \cmidrule{2-5}
& unrestricted & \makecell{Bell sampling and gentle\\ measurements \cite{huang2021information}}& $O(n/\eps^4)$ & \checkmark\\
\bottomrule

\end{tabular}
\caption{A summary of Pauli shadow tomography algorithms for the sets of observables we study. Cells in blue are new results in this paper. \\ *The constant depends on $k$. \\ $^\dag$The degree of the polynomial depends on $k$.\\
$^\ddagger$Based on single-copy measurements in uniformly random Clifford bases.}
\label{table}
\end{table}

There are many questions left open by our work. Is it possible to improve the upper bounds from \cref{eq:chifermi_} and \cref{eq:chign}---e.g., can we establish better chi-binding functions for the (families of) commutation graphs of interest? Is rapid-retrieval compression possible for smaller error parameters, e.g., $\eps=1/\mathrm{poly}(n)$? Can we devise triply efficient learning algorithms for \text{any subset} of Pauli observables?

One route towards resolving these questions would be via improved algorithms for coloring the commutation graph $G(S_\eps)$. The following conjecture asserts an efficient fractional coloring of the commutation graph of any subset of Paulis that has simultaneously large expected values in a quantum state.

\begin{conjecture} \label{conj:coloring_intro}
Let $\rho$ be an $n$-qubit state, $\delta\in (0,1)$, and let $B\subseteq \mathcal{P}^{(n)}$ be the set of all Paulis $P$ such that $|\Tr(\rho P)| \geq \delta$. There is a fractional coloring of the commutation graph $G(B)$ of size $O(1/\delta^{2})$. 
\end{conjecture}

If this conjecture holds, and in addition the fractional coloring is suitably efficient,\footnote{In particular, we require that a fractional coloring of size $O(1/\delta^{2})$ for any subset $R\subseteq B$ can be sampled in time $\mathrm{poly}(|R|,n,1/\delta)$} then we would obtain a triply efficient Pauli shadow tomography algorithm for \textit{any} subset $S$ of Pauli observables. Moreover, the learning algorithm would also output a rapid-retrieval Pauli compression of all observables in $S$, of size $O(n^2(\log |S|)/\eps^4)$, see \Cref{sec:two_copy_learning}.

To address \Cref{conj:coloring_intro}, it is natural to ask if  \Cref{lem:clique_bound_informal} can be strengthened by showing that $O(1/\epsilon^{2})$ is in fact an upper bound on the \textit{fractional clique number}---a well-known linear programming relaxation of the clique number~\cite{scheinerman2011fractional}. The existence of a suitable fractional coloring stated in \Cref{conj:coloring_intro} would then follow from linear programming duality, which asserts that the fractional clique number of any graph equals its fractional chromatic number.

The remainder of the paper is organized as follows. In \Cref{sec:commutation} we describe the commutation graph and its properties, as well as our frameworks for one- and two-copy Clifford learning. In \Cref{sec:learnallpaulis} we describe techniques that go into our algorithms for learning all Paulis: in \Cref{sec:signs_from_magnitudes} we describe the algorithm to compute a mimicking state, which completes the proof of \Cref{thm:mmw_informal}, and in \Cref{sec:coloring} we give the proof of \Cref{lem:coloring_intro}. Finally in \Cref{sec:majoranas} we establish our results concerning shadow tomography for $k$-body fermionic operators. In \Cref{sec:lower_bound} we establish the lower bound \Cref{thm:notfermions} for single-copy learning. In \Cref{sec:1-RDM} we describe the simple triply efficient learning algorithm for the special case $k=1$, and then in \Cref{sec:kbody} we consider the case $k\geq 2$ and prove \Cref{lem:chibindfermion_intro}.

\paragraph{Independent work.} While writing this paper, we learned of related, independent work by Chen, Gong, and Ye~\cite{Chen2024} that studies the sample complexity of Pauli shadow tomography with limited quantum memory. We have coordinated our arXiv submissions so that both papers appear on the same day.

\section*{Acknowledgments}
We thank Bill Huggins for valuable comments on this manuscript. 
DG thanks Jim Geelen for a helpful discussion about chi-boundedness, and Sophie Spirkl for explaining how the result of Gy\'arf\'as \cite{gyarfas1987problems} can be turned into an algorithm. RB thanks Nicholas Rubin and Joonho Lee for discussions about using fermionic RDMs in embedding and perturbation theory contexts. DG is a CIFAR fellow in the quantum information science program.

\section{Commutation, learning, and coloring} \label{sec:commutation}

In this section we describe properties of the commutation graph $G(S)$ of a set of Pauli observables $S \subseteq \mathcal{P}^{(n)}$, and the connection between these properties and shadow tomography algorithms. In \Cref{sec:commutation_index} we describe the tension between learning and anticommutativity of a set of observables. This is quantified by a \textit{commutation index} which was shown in Ref.~\cite{chen2022exponential} to lower bound the sample complexity of single-copy learning. Then in Sections \ref{sec:fractional_colorings} and \ref{sec:two_copy_learning} we describe our frameworks for one- and two-copy learning with probabilistic Clifford measurement strategies.

\subsection{Commutation index}\label{sec:commutation_index}

The following result represents a kind of uncertainty principle, generalizing the familiar Bloch sphere constraint $\langle X\rangle^2 + \langle Y\rangle^2 + \langle Z\rangle^2 \leq 1$.

\begin{lemma} \label{lem:anticommuting_bound}\label{lem:pairwise}
If Pauli operators $P_1, \dots, P_m$ pairwise anticommute, then for any state $\rho$
\begin{equation}
\sum_j \Tr\left(P_j \rho\right)^2 \leq 1.
\end{equation}
\end{lemma}

Versions of \Cref{lem:anticommuting_bound} appear in various papers, for example \cite[Theorem 1]{asadian2016heisenberg}. Since the proof is simple, we reproduce it here.

\begin{proof}
Given $\rho$, let $a_j = \Tr\left(P_j \rho\right)$. We aim to show $\sum_j a_j^2 \leq 1$.
Consider the observable
\begin{equation}
Q = \sum_j a_j P_j.
\end{equation}
We will use the inequality $\Tr(Q^2\rho) - \Tr(Q\rho)^2 = \text{Var}_\rho(Q) \geq 0$. Formally, this holds since $\rho$ is positive semi-definite so $\Tr(O\rho) \geq 0$ for any positive semi-definite operator $O$ and
\begin{equation}
\Tr\left(\big(Q - \Tr\left(Q \rho\right)  \mathbbm{1}\big)^2 \rho\right) \geq 0 \
\implies \ \Tr\left(Q \rho\right)^2 \leq \Tr\left(Q^2 \rho\right).
\end{equation}
Due to anticommutativity of $P_1, \dots, P_m$, we have
\begin{equation}
Q^2 = \sum_j a_j^2 P_j^2 + \sum_{j \neq l} a_j a_l P_j P_l
= \sum_j a_j^2\cdot  \mathbbm{1} + \frac{1}{2} \sum_{j \neq l} a_j a_l \{P_j,P_l\}
= \sum_j a_j^2\cdot  \mathbbm{1}.
\end{equation}
Note $P_j^2 =  \mathbbm{1}$ since they are Hermitian unitaries. Thus $\Tr(Q^2\rho) = \sum_j a_j^2$. On the other hand $\Tr(Q\rho) = \sum_j a_j^2$, thus
\begin{equation}
\Tr\left(Q \rho\right)^2 \leq \Tr\left(Q^2 \rho\right) \ \implies \ \Big(\sum_j a_j^2\Big)^2 \leq \sum_j a_j^2 \ \implies \ \sum_j a_j^2 \leq 1.\qedhere
\end{equation}
\end{proof}

In \Cref{lem:anticommuting_bound}, we saw that pairwise anticommuting Pauli operators cannot all have large expected values in a quantum state, as the sum of their squared expected values cannot exceed $1$. More generally, the average of squares of expected values of Paulis in a set $S$ is defined to be its \textit{commutation index}.

\begin{definition} \label{def:range}
For a set $S$ of Pauli operators, define their commutation index by
\begin{equation}
\Delta(S) = \frac{1}{|S|}\max_{\rho} \sum_{P \in S} \Tr\left(P \rho\right)^2.
\end{equation}
\end{definition}
For example, \Cref{lem:anticommuting_bound} shows that if all the Paulis in $S$ anticommute, then $\Delta(S) \leq 1/|S|$.

Ref.~\cite{chen2022exponential} shows that the inverse of the commutation index is a lower bound on the sample complexity of single-copy shadow tomography for $S$. It can be interpreted as describing a tension between anticommutation and learnability. (Maximization over states as stated in \cite[Equation 79]{chen2022exponential} can be replaced by maximization over density matrices via convexity.)

\begin{theorem}[Theorem 5.5, \cite{chen2022exponential}] \label{thm:CCHL22}
Shadow tomography to precision $\eps$ for a set $S$ of Pauli observables with single-copy measurements requires at least
\begin{equation}
\Omega\left(\frac{1}{\eps^2 \Delta(S)}\right)
\end{equation}
copies of $\rho$. This holds even for adaptive measurement strategies.\footnote{The lower bound holds even if the shadow tomography scheme is only able to output the absolute values of the expectation values to precision $\eps$.}
\end{theorem}

The commutation index $\Delta(S)$ can be upper bounded in terms of the \emph{Lovasz $\vartheta$-function} \cite{de2023uncertainty, xu2023bounding, hastings2022optimizing}.

\begin{definition}
Let $G$ be a graph on $m$ vertices. The \emph{Lovasz $\vartheta$-function} $\vartheta(G)$ is defined by the following semidefinite program of dimension $m$: 
\begin{equation}
\max \ \{ \Tr\left(J X\right), \ X \in \mathbb{R}^{m \times m} \ \text{s.t.} \ X \succeq 0 , \ \Tr{X} = 1 , \ X_{jl} = 0 \ \forall (j,l) \in E  \}, \label{eq:theta_primal}
\end{equation}
where $E$ denotes the edge set of the graph $G$ and $J$ the all-ones matrix. 
It has dual
\begin{equation}
\min \ \{ \lambda \in \mathbb{R} \ \text{s.t.} \ \exists \ A \in \mathbb{R}^{m \times m} , \ A_{jj} = 1 \ \forall j , \ A_{jl} = 0 \ \forall (j,l) \notin E , \ \lambda A \succeq J  \}. \label{eq:theta_dual}
\end{equation}
\end{definition}

The following result is a generalization of \Cref{lem:anticommuting_bound}, with proof in \Cref{app:commutation_graph}.

\begin{restatable}{lemma}{commutationgraph}\label{lem:commutation_graph}
\emph{(\cite{de2023uncertainty, xu2023bounding, hastings2022optimizing})}
Let $S$ be a set of Pauli operators. Then
\begin{equation}
\Delta(S) \leq \frac{\vartheta(G(S))}{|S|}.
\end{equation}
\end{restatable}

\subsection{Fractional coloring and single-copy Clifford learning} \label{sec:fractional_colorings}

Here we describe the connection between shadow tomography algorithms that learn Pauli observables $S\subseteq \mathcal{P}^{(n)}$ using  probabilistic Clifford measurements, and fractional colorings of the commutation graph $G(S)$. A fractional coloring is a relaxation of the usual notion of graph coloring \cite{scheinerman2011fractional}.

\begin{definition} \label{def:fractional_coloring}
Let $G=(V,E)$ be a graph. A fractional coloring of $G$ of size $\chi$ is a probability distribution $q$ over independent sets $I \subseteq V$ with the property that
\begin{equation}
\forall  v \in V: \quad \mathrm{Pr}_{I \sim q}(v \in I) \geq 1 / \chi.
\end{equation}
The fractional chromatic number $\chi_f(G)$ of $G$ is the size of the smallest fractional coloring of $G$.
\end{definition}
Note that the size $\chi$ of a fractional coloring need not be an integer. Also note that a (standard, non-fractional) coloring of $G$ with $\chi$ colors can be regarded as a fractional coloring of size $\chi$, corresponding to a uniform distribution over color classes. So the fractional chromatic number of a graph is upper bounded by its chromatic number.

\Cref{thm:fractional}, restated below, asserts that if we have a fractional coloring of the commutation graph $G(S)$ of small size, then we can learn the expectation values of all observables in $S$ with few single-copy Clifford measurements. In particular, the sample complexity of the algorithm scales linearly with the size of the fractional coloring. 

In the following, samples from the fractional coloring are represented as binary vectors of length $|S|$ whose support is an independent set in $G(S)$. In this setting, the runtime to produce a single sample from a fractional coloring always satisfies $T\geq |S|$.

\fractional*

\begin{proof}
An independent set $I$ in the commutation graph $G(S)$ consists of a set of mutually commuting Pauli operators, which can be simultaneously measured by applying a Clifford circuit and measuring in the computational basis. Such a Clifford circuit can be computed using a classical algorithm with $O(n^3)$ runtime, via standard techniques in the stabilizer formalism \cite{aaronson2004improved}.

If we draw an independent set $I$ from a fractional coloring $q$ with size $\chi$ and then measure $\rho$ in the corresponding Clifford basis, the result gives us a measurement of Pauli $P\in S$ whenever $P\in I$. Note that it is also possible that we get more useful measurements than this---the Clifford unitary may diagonalize some Paulis $P$ that are not in $I$. 

Suppose we repeat this process independently $N$ times, sampling independent sets $I_1,I_2,\ldots, I_N$ and measuring $N$ independent identical copies of $\rho$ in the corresponding Clifford bases $C_1,C_2,\ldots, C_N$. For each $P\in S$, let $x_P^{j}\in \{-1,0,1\}$ be the random variable that is equal to the measured outcome of $P$ if it is diagonalized by $C_j$, and zero otherwise. 

We have
\begin{equation}
\mathrm{Pr}[x_P^{j}\in \{-1,1\}]\geq \mathrm{Pr}[P\in I_j]\geq  \frac{1}{\chi} \qquad j\in [N
] \qquad P\in S,\label{eq:probdiagonal}
\end{equation}
where we used the fact that $q$ is a fractional coloring of size $\chi$.

For each Pauli $P\in S$, let
\begin{equation}
R_P=\{j: x_P^j\in \{-1,1\}\} \quad \text{and} \quad 
N_P=|R_P|.
\end{equation}
Consider the sample mean
\begin{equation}
\tilde{P}= \frac{1}{N_P}\sum_{j\in R_P} x_P^{j}.
\end{equation}
Conditioned on a fixed value $N_P\geq 1$, this sample mean $\tilde{P}$ is an average of $N_P$ independent $\pm 1$-valued random variables. It satisfies 
\begin{equation}
\mathbb{E}(\tilde{P})=\mathrm{Tr}(\rho P) \qquad \text{and} \qquad \mathrm{Var}(\tilde{P})\leq \frac{1}{N_P}.
\end{equation}
By Chebyshev's inequality we have
\begin{equation}
\mathrm{Pr}\left[|\tilde{P}-\mathrm{Tr}(\rho P) |\geq \eps \; \; \big| \; \;  N_P\geq 100/\epsilon^2\right]\leq 0.01. 
\end{equation}
From Eq.~\eqref{eq:probdiagonal} we see that by taking $N=O(\chi / \eps^{2})$ we can ensure that, for a given Pauli $P$, we have $N_P\geq 100/\epsilon^2$ with probability at least $0.99$ (say). Therefore, 
\begin{align}
\mathrm{Pr}\left[|\tilde{P}-\mathrm{Tr}(\rho P) |\leq \eps\right]&\geq \mathrm{Pr}\left[N_P\geq 100/\epsilon^2\right]\cdot \mathrm{Pr}\left[|\tilde{P}-\mathrm{Tr}(\rho P) |\leq \eps \; \; \big| \; \;  N_P\geq 100/\epsilon^2\right]\\
&\geq 0.99^2
\end{align}
for each Pauli $P\in S$.

Now repeat the above process $L$ times, generating sample means $\tilde{P}_1,\ldots, \tilde{P}_L$ for each $P\in S$, and consider the median-of-means estimator
\begin{equation}
\lambda_P= \mathrm{median}(\tilde{P}_1,\tilde{P}_2,\ldots, \tilde{P}_L).
\label{eq:medmeans}
\end{equation}
By choosing $L=O(\log |S|)$ we can ensure that, for each $P\in S$ we have $|\lambda_P-\mathrm{Tr}(\rho P)|\leq \eps$ with probability at least $0.01/|S|$. By a union bound we get all the expected values in $S$ to within $\eps$ with probability at least $0.99$. The total number of samples of $\rho$ and the total number of samples from the fractional coloring $q$ used in the algorithm, are both at most $LN=O(\chi (\log|S|)/\eps^{2})$.

Now consider the runtime of the protocol. The independent sets $I$ in the fractional coloring are specified explicitly as subsets of $S$, so the median-of-means estimator $\lambda_P$ can be computed for all $P\in S$ with a runtime $O(LN|S|)$ once we have already obtained all the measurement data $\{x_P^j\}$. The total runtime is therefore upper bounded as $O((T+n^3+|S|)LN)$, where the first term is the cost of sampling the fractional colorings, the second term is the cost of computing a Clifford circuit for each sample (and applying this circuit to measure the state), and the third term is the cost of postprocessing. Since $T\geq |S|$  the runtime simplifies to $O((T+n^3)LN)$.
\end{proof}

Finally, let us show that single-copy measurement strategies  from \Cref{thm:fractional} learn a \textit{compressed classical representation} of $\rho$
that encodes the expected values of all Pauli observables from $S$ (to within error $\epsilon$).
 Indeed, each Clifford measurement basis has an efficient classical description consisting of $O(n^2)$ bits. We can imagine a version of the learning algorithm described in \Cref{thm:fractional} where, after the measurements are performed using \cref{eq:numfrac} copies of $\rho$, the resulting measurement outcomes and measurement bases are packaged up into a classical description $D(\rho, S, \eps)$ of size
\begin{equation}
O(n^2\chi(\log{|S|})/\eps^2).
\end{equation}
Here there is a factor of $n^2$ for each measurement basis (each measurement outcome only requires $n$ bits to describe and so describing the outcomes requires asymptotically fewer bits than describing the bases). The efficient protocol for extracting expected values $\mathrm{Tr}(\rho P)$ with $P\in S$ (up to $\eps$ error) can be performed using only the compressed classical description $D(\rho, S, \eps)$.

\subsection{Cliques and two-copy Clifford learning} \label{sec:two_copy_learning}

As discussed in the Introduction, we propose a framework for two-copy learning that uses an initial stage of Bell sampling (as in Ref.~\cite{huang2021information}) to determine a set $S_{\eps}\subseteq S$ that with high probability satisfies \cref{eq:seps}, which we restate:
\begin{equation}\label{eq:eps_support}
    |\Tr(\rho P)| \geq \eps/2 \quad \text{for all} \quad P\in S_{\eps}.
\end{equation}
This step uses $O((\log|S|)/\eps^{4})$ Clifford measurements on copies of $\rho\otimes \rho$. Then we aim to learn all observables in $S_{\eps}$ using the fractional coloring approach described in the previous section. In particular, we aim to find an (efficiently sampleable) fractional coloring of the commutation graph $G(S_{\eps})$. 

\begin{lemma}[Template for two-copy Clifford shadow tomography]
\label{lem:twocopycolor}
Suppose that $G(S_\epsilon)$ admits a fractional coloring of size $\chi$ that can be sampled by a randomized algorithm with runtime $T$. Then there is an algorithm which performs shadow tomography for $S$ using 
\begin{equation}
O((\log|S|)/\eps^{4}+\chi(\log|S|)/\eps^{2})
\end{equation}
two-copy Clifford measurements and runtime $O(|S|(\log|S|)/\eps^{4}+(T+n^3)\chi(\log|S|)/\eps^{2})$. 
\end{lemma}
\begin{proof}
The first step uses $O((\log{|S|})/\epsilon^{4})$ two-copy Bell measurements and classical runtime $O(|S|(\log{|S|})/\epsilon^{4})$  to compute the set $S_{\epsilon}$. We output zero as our estimate for the expected value of any Pauli  in $S\setminus S_{\epsilon}$. Then we use the single-copy learning protocol from \Cref{thm:fractional} to compute estimates of $\mathrm{Tr}(\rho P)$ for all $P\in S_{\epsilon}$. This second step uses runtime $O((T+n^3)  \chi (\log{|S|})/\epsilon^{2} )$.  
\end{proof}

The protocol from \Cref{lem:twocopycolor} is based on Clifford measurements, which have an efficient classical description. Because of this, and the fact that the classical postprocessing is simple, we obtain the following classical compressed representation of $\rho$.

\begin{lemma}[Rapid-retrieval compression]
\label{lem:rapidretrieval}
The shadow tomography protocol described in \Cref{lem:twocopycolor} learns a compressed classical representation of $\rho$ consisting of $B$ bits, where
\begin{equation}
B=O(n(\log|S|)/\eps^{4}+n^2\chi(\log|S|)/\eps^{2}).
\label{eq:bits}
\end{equation}
With high probability, this compressed representation has the following rapid-retrieval property. There is a classical algorithm which, given this classical data and any Pauli $P\in S$, outputs an estimate of $\mathrm{Tr}(\rho P)$ to within $\epsilon$ error. The runtime of the algorithm is $O(B)$.
\end{lemma}

\begin{proof}
The compressed representation consists of the $O((\log|S|)/\epsilon^{4})$ Bell samples (each one is $2n$ bits) as well as a classical description of each Clifford measurement basis and measurement outcome used in the second stage of the learning protocol. That is,  $N=O(\chi (\log|S|)/\epsilon^{2})$ Clifford measurement bases $C_1,C_2,\ldots, C_N$ (each described by a circuit with $O(n^2)$ one- and two-qubit Clifford gates) and corresponding measurement outcomes $z^1,z^2,\ldots, z^N \in \{0,1\}^{n}$. Given a Pauli $P\in S$ and this classical data, we can compute an $\epsilon$-error estimate of $\mathrm{Tr}(\rho P)$ in the following way. First, using the Bell samples, we determine if $P\in S_{\epsilon}$. This step requires us to compute a sample mean over the Bell samples, using runtime $O(n(\log|S|)/\epsilon^{4})$. If $P\notin S_{\epsilon}$, we output $0$ as our estimate. If $P\in S_{\epsilon}$ then we compute the median-of-means estimator from \cref{eq:medmeans}. To do this we have to compute indicator functions $x^j_P\in \{1,0,-1\}$ that describe the measured outcome of Pauli $P$ for each Clifford measurement basis $j$, which is given by
\begin{equation}
x_P^{j}=\langle z^{j}|C_j^{\dagger}P C_j|z^j\rangle.
\end{equation}
The RHS is computed using the stabilizer formalism: we update the Pauli $P$ by conjugating each gate in the circuit $C_j$ one-by-one, and this process takes a total runtime $O(n^2)$ since there are $O(n^2)$ one- and two-qubit Clifford gates in the circuit. The total runtime to extract the estimate of $\mathrm{Tr}(\rho P)$ is therefore
\begin{equation}
O(n(\log|S|)/\eps^{4}+n^2 \chi (\log|S|)/\eps^{2}).\qedhere
\end{equation}
\end{proof}

\Cref{lem:twocopycolor} forms the basis of several of the two-copy protocols that we present in this work. To use this framework one needs to find an efficiently sampleable fractional coloring of $G(S_{\epsilon})$. 

A challenge here is that the set $S_{\epsilon}$ and its commutation graph depend in a potentially complicated way on the unknown state $\rho$.  Ideally, we would like to understand any structural properties of this graph that can be leveraged to compute good fractional colorings.   In this paper we we will only exploit two simple properties:  (A) $G(S_{\epsilon})$ is an induced subgraph of $G(S)$ and (B) with high probability, $G(S_{\epsilon})$ does not have large cliques, as described in \Cref{lem:clique_bound_informal}, which we restate and prove below. 

\largestclique*
\begin{proof}
Recall that $S_{\epsilon}$ satisfies \cref{eq:eps_support} with high probability. We show that in this case the largest clique in  $G(S_{\epsilon})$ has size at most $4/\epsilon^2$.

Suppose there is a clique in $G(S_{\epsilon})$ of size $\omega$. The vertices of the clique are a set of pairwise anticommuting Pauli operators $P_1, P_2,\ldots, P_\omega$. Applying \Cref{lem:pairwise} with these operators and the state $\rho$ gives
\begin{equation}
\sum_{j=1}^{\omega} \mathrm{Tr}(P_j \rho)^2\leq 1.
\end{equation}
On the other hand from \cref{eq:eps_support} we have $\mathrm{Tr}(P_j \rho)^2\geq \epsilon^2/4$ for each $1\leq j\leq \omega$. Plugging into the above gives $\omega \eps^2/4\leq 1$, and therefore the size of the maximal clique is upper bounded as $\omega \leq 4/\epsilon^{2}$.
\end{proof}
To use our framework to learn Pauli observables $S\subseteq \mathcal{P}^{(n)}$, it suffices to establish a so-called \textit{chi-binding function} for the family of induced subgraphs of $G(S)$. That is, we seek a function $g(\omega)$ such that:

\vspace{0.2cm}
\textit{For any induced subgraph $G'$ of $G(S)$ with largest clique of size $\omega$, there is a fractional coloring of $G'$ with size $\chi\leq g(\omega)$.}
\vspace{0.2cm}

A statement of this form implies---via \Cref{lem:twocopycolor} and \Cref{lem:clique_bound_informal}---a shadow tomography algorithm that uses two-copy Clifford measurements and has sample complexity 
\begin{equation}
O((\log|S|)/\eps^{4}+g(4/\epsilon^2)(\log|S|)/\eps^{2}).
\end{equation}
\Cref{lem:twocopycolor} also gives an upper bound on the runtime of the protocol in terms of the time required to sample from the fractional coloring.

\section{Learning all Pauli observables} 
\label{sec:learnallpaulis}
In \Cref{sec:signs_from_magnitudes} we provide further details of the shadow tomography algorithm for all Paulis stated in \Cref{thm:mmw_informal}. Then, in \Cref{sec:coloring}, we give the proof of \Cref{lem:coloring_intro} which is used to establish the rapid retrieval Pauli compression stated in \Cref{cor:compression_intro}.

\subsection{Computing a mimicking state}
\label{sec:signs_from_magnitudes}

Here we complete the proof of \Cref{thm:mmw_informal} by showing that a mimicking state can be computed using $O(n\log(n/\epsilon)/\epsilon^{4})$ samples of $\rho$ and runtime $\mathrm{poly}(2^n, 1/\epsilon)$.

\begin{algorithm}[H] \caption{Compute a mimicking state} \label{alg:hypo_state}

\textbf{Input:} A precision parameter $\eps\in (0,1)$, $O(n \log(n/\eps) / \eps^4)$ copies of an unknown $n$-qubit state $\rho$, and estimates $\{u_P\}_{P\in \mathcal{P}^{(n)}}$, such that 
\begin{equation}
u_P \geq 0 \quad \text{and} \quad |u_P - |\Tr(P \rho)|| \leq \eps/4  \qquad P \in \mathcal{P}^{(n)}.
\end{equation}

\textbf{Output:} A classical description of a density matrix $\sigma$ satisfying the following mimicking state condition with high probability:
\begin{equation}
\label{eq:outputcondition}
|\Tr(P \sigma)| \geq \eps/4 \quad \text{for all} \ P \in \mathcal{P}^{(n)} \ \text{with} \ u_P \geq 3\eps/4.
\end{equation}

\textbf{Algorithm:}
\begin{enumerate}
    \item Set $T = \lceil 64 n / \eps^2\rceil$+1 and $\beta = \sqrt{n / T}$.
    \item Initialize $\omega^{(0)} =  \mathbbm{1} / 2^n$, the maximally mixed state.
    \item For $t = 0,\dots,T-1$,
    \begin{enumerate}
        \item Search for Pauli $P\in \mathcal{P}^{(n)}$ such that $u_P \geq 3\eps/4$ \emph{and} $|\Tr(P \omega^{(t)}) - u_P| > \eps/2$ \emph{and} $|\Tr(P \omega^{(t)}) + u_P| > \eps/2$.
        \item If there is no such Pauli $P$ then we are done, since we are guaranteed $|\Tr(P \omega^{(t)})| \geq \eps/4$ whenever $u_P \geq 3\eps/4$. Output $\sigma = \omega^{(t)}$.
        \item Else, do:
        \begin{enumerate}
            \item Use $O((\log T)/\eps^2)$ copies of $\rho$ to compute an estimate $r_P \in \{+1,-1\}$ such that $r_P=\sign(\mathrm{Tr}(\rho P))$ with probability at least $1-0.01/|T|$.
            \item Set
            \begin{equation}
            M^{(t)} = \sign\{\Tr(P \omega^{(t)}) - r_P u_P\} \cdot P.
            \end{equation}
            \item Set
            \begin{equation} \label{eq:mmw_gibbs}
            \omega^{(t+1)} = \frac{\exp\left(- \beta \sum_{\tau=1}^t M^{(\tau)}\right)}{\Tr\left(\exp\left(- \beta \sum_{\tau=1}^t M^{(\tau)}\right)\right)}.
            \end{equation}
        \end{enumerate}
    \end{enumerate}
    \item Output $\sigma = \omega^{(T)}$.
\end{enumerate}
\end{algorithm}

Before showing correctness of the algorithm, let us discuss its runtime and sample complexity. During the course of the algorithm we store a classical representation of the state $\omega^{(t)}$ as a matrix of size $2^n\times 2^n$, and in each of the $T=O(n/\eps^2)$ steps we need to exhaustively search over the set of Paulis $\mathcal{P}^{(n)}$, compute Pauli expected values in the state $\omega^{(t)}$, and compute matrix exponentials \cref{eq:mmw_gibbs}. The runtime of the algorithm is polynomial in the Hilbert space dimension $2^n$ and $1/\eps$. The sample complexity is $O(T (\log T)/\eps^2)$.

\Cref{alg:hypo_state} is a variant of the \emph{matrix multiplicative weights} algorithm, which has been used previously in the context of quantum learning \cite{brandao2017quantum,aaronson2018online}. To show correctness we follow the standard analysis, with some modifications.

\begin{lemma} \label{lem:MMW}
\emph{(\cite{arora2007combinatorial} Theorem 3.1)}
Suppose \Cref{alg:hypo_state} reaches step 4. Then the states \cref{eq:mmw_gibbs} satisfy
\begin{equation} \label{eq:regret_bound}
\sum_{t=1}^T \Tr(M^{(t)} \omega^{(t)}) - \lambda_{\min}\left(\sum_{t=1}^T M^{(t)}\right) \leq 2 \sqrt{nT}.
\end{equation}
\end{lemma}

Here $\lambda_{\min}(\cdot )$ denotes the smallest eigenvalue.
\begin{lemma}
With high probability the output of \Cref{alg:hypo_state} satisfies \cref{eq:outputcondition}.
\end{lemma}
\begin{proof}

Let $P^{(t)}$ be the Pauli found at step $t$, and define
\begin{equation}
y^{(t)}= \Tr(P^{(t)} \omega^{(t)}) \ , \ r^{(t)} = r_{P^{(t)}} \ , \ u^{(t)}= u_{P^{(t)}}.
\end{equation}

By a union bound, with probability at least $0.99$, all estimates $r^{(t)}$ computed during the course of the algorithm satisfy $r^{(t)}=\sign(\mathrm{Tr}(\rho P^{(t)}))$. In the following we show that in this case the output of the algorithm is guaranteed to satisfy \cref{eq:outputcondition}.

Note that if the algorithm outputs a state $\sigma=\omega^{(t)}$ for some $t<T$ then the algorithm is correct, as guaranteed by the output condition. To complete the proof, below we show that this is the case---the algorithm never reaches step 4. 

Toward a contradiction, suppose the algorithm does reach step 4. By \Cref{lem:MMW}, we have
\begin{align}
\sum_{t=1}^T \sign\{y^{(t)} - r^{(t)} u^{(t)}\} \cdot y^{(t)} &\leq \lambda_{\min}\left(\sum_{t=1}^T \sign\{y^{(t)} - r^{(t)} u^{(t)}\} \cdot P^{(t)}\right) + 2 \sqrt{nT} \\
&\leq \sum_{t=1}^T \sign\{y^{(t)} - r^{(t)} u^{(t)}\} \cdot \Tr(P^{(t)} \rho) + 2 \sqrt{nT}.
\end{align}
In the last equation, we substituted our state $\rho$. Equivalently, we have
\begin{equation} \label{eq:regret_implication}
\sum_{t=1}^T \sign\{y^{(t)} - r^{(t)} u^{(t)}\} \cdot \Big(y^{(t)} - \Tr(P^{(t)} \rho)\Big) \leq 2 \sqrt{nT}.
\end{equation}

For each $t$ we have $|y^{(t)} - r^{(t)} u^{(t)}| > \eps/2$. But also $|r^{(t)} u^{(t)} - \Tr(P^{(t)} \rho)| \leq \eps/4$. Together these conditions imply
\begin{equation}
\label{eq:signs}
\sign\{y^{(t)} - r^{(t)} u^{(t)}\} = \sign\{y^{(t)} - \Tr(P^{(t)} \rho)\}
\end{equation}
and
\begin{equation}
\left|y^{(t)} - \Tr(P^{(t)} \rho)\right| > \frac{\eps}{4} \label{eq:online2}.
\end{equation}

From \cref{eq:signs} we get
\begin{equation}
\sign\{y^{(t)} - r^{(t)} u^{(t)}\} \cdot \Big(y^{(t)} - \Tr(P^{(t)} \rho)\Big) = \Big|y^{(t)} - \Tr(P^{(t)} \rho)\Big| \label{eq:online1}
\end{equation}
and combining this with \cref{eq:regret_implication} we get
\begin{equation}
\sum_{t=1}^T \left|y^{(t)} - \Tr(P^{(t)} \rho)\right| \leq 2 \sqrt{nT}.
\end{equation}
Plugging in \cref{eq:online2}  gives
\begin{equation}
T \cdot \frac{\eps}{4} \leq 2 \sqrt{nT}, 
\end{equation}
and therefore $T\leq 64n/\eps^{2}$. This is a contradiction, since  $T=\lceil 64n/\eps^{2}\rceil+1$ is larger than this. So the algorithm never reaches step 4, as claimed.
\end{proof}

\subsection{Coloring commutation graphs with bounded clique number}
\label{sec:coloring}

In this section we prove \Cref{lem:coloring_intro}, restated below.

\chibind*

We shall use the following algorithmic version of a result of Gy\'arf\'as, which we prove below.

\begin{lemma}[Algorithmic version of Thm.~2.4  of \cite{gyarfas1987problems}]
\label{lem:chi_binding}
Suppose $G$ is a graph on $m$ vertices whose longest induced path has $\ell$ vertices, with $\ell \geq 1$, and clique number $\omega$. Then there is a classical algorithm which colors $G$ using $\ell^{\omega - 1}$ colors and runtime $O(m^2 \omega)$.
\end{lemma}

\begin{proof}[Proof of \Cref{lem:coloring_intro}.]
Let $G'$ be an induced subgraph of $G(\mathcal{P}^{(n)})$. Below we show that $G(\mathcal{P}^{(n)})$, and therefore also $G'$, does not contain any induced paths with more than $2n + 1$ vertices. The claim then follows by applying \cref{lem:chi_binding}.

Suppose $P_1,P_2,\ldots,P_s$ is an induced path in $G(\mathcal{P}^{(n)})$, i.e.,
\begin{equation}
P_iP_{i+1}=-P_{i+1}P_i \text{ for }\quad 1\leq i\leq s-1 \quad \text{and} \quad [P_i, P_j]=0 \quad |i-j|\geq 2.
\label{eq:commute}
\end{equation}
Define Pauli operators
\begin{equation}
Q_r=P_1P_2\ldots P_r  \qquad 1\leq r\leq s.
\end{equation}
We now use \cref{eq:commute} to show that these operators are pairwise anticommuting. To see this note that
\begin{equation}
Q_{r+a}=\left(P_1P_2\ldots P_r\right)\left( P_{r+1}P_{r+2}\ldots P_{r+a}\right)=-\left( P_{r+1}P_{r+2}\ldots P_{r+a}\right)\left(P_1P_2\ldots P_r\right),
\end{equation}
where we used $P_r P_{r+1}=-P_{r+1}P_r$ and the fact that $[P_i, P_j]=0$ whenever $i\leq r-1$ and $j\geq r+1$. Therefore
\begin{equation}
Q_r Q_{r+a}=-Q_r\left( P_{r+1}P_{r+2}\ldots P_{r+a}\right)\left(P_1P_2\ldots P_r\right)=-Q_{r+a}Q_r,
\end{equation}
which shows that $\{Q_j\}_{j\in [s]}$ are pairwise anticommuting Pauli operators. 

It is a well known fact that the $n$ qubit Hilbert space does not contain any set of pairwise anticommuting Pauli operators with size greater than $2n+1$ (see for example Appendix G of Ref.~\cite{bonet2020nearly}). Thus $s \leq 2n+1$.
\end{proof}

\begin{figure}[thb]
	\centering
	\begin{subfigure}[t]{0.23\textwidth}
		\centering   
		\begin{tikzpicture}[every node/.style={circle,draw}]
			\node (1) at (90:1.5cm) {1};
			\node (2) at (162:1.5cm) {2};
			\node (3) at (234:1.5cm) {3};
			\node (4) at (306:1.5cm) {4};
			\node (5) at (18:1.5cm) {5};
			    
			\draw (1) -- (2);
			\draw (2) -- (3);
			\draw (3) -- (4);
			\draw (4) -- (5);
			\draw (5) -- (1);
		\end{tikzpicture}		           
		\caption{Original graph}
	\end{subfigure}
	~        
	\begin{subfigure}[t]{0.23\textwidth}
		\centering
		\begin{tikzpicture}
			[sibling distance=5em,
				every node/.style = {shape=circle, draw, align=center},
			level distance=1.25cm]			        
			\node {1}
			child { node {2}
				child { node {3} } }
			child { node {5}
				child { node {4} } };
		\end{tikzpicture}		        
		\caption{BFS spanning tree}
	\end{subfigure}
	~ 
	\begin{subfigure}[t]{0.23\textwidth}
		\centering
		\begin{tikzpicture}
			[every node/.style={shape=circle, draw, align=center},
			level distance=1cm]
			  
			\node {1}
			child { node {2} 
				child { node {3}
					child { node {4}
						child { node {5} }
					}
				}
			};
		\end{tikzpicture}		        
		\caption{DFS spanning tree}
	\end{subfigure}
	~ 
	\begin{subfigure}[t]{0.23\textwidth}
		\centering
		\begin{tikzpicture}
			[every node/.style={shape=circle, draw, align=center},
			level distance=1cm]			  
			\node {1}
			child { node {2} 
				child { node {3}
					child { node {4} }
				}
			}
			child { node {5}
			};
		\end{tikzpicture}		        
		\caption{NFS spanning tree}
	\end{subfigure}	    
	\caption{An example showing the spanning trees generated by breadth-first search (BFS), depth-first search (DFS), and our algorithm neighbour-first search (NFS) for the cycle on 5 vertices.}
\end{figure}
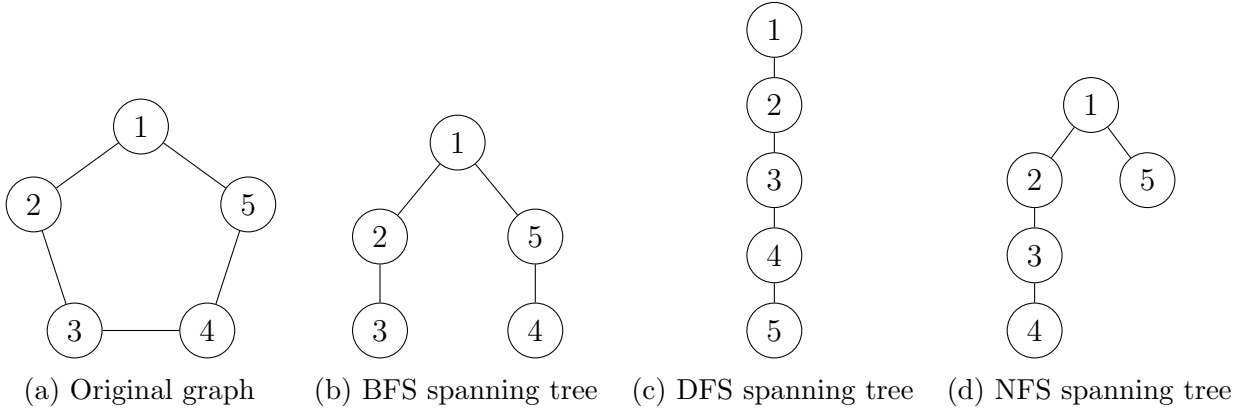

The algorithm of \Cref{lem:chi_binding} relies on a non-standard graph traversal algorithm, which can be interpreted as a combination of depth-first-search and breadth-first-search. The graph search algorithm begins with an arbitrary seed vertex $v$ in $G$, and generates a spanning tree of $G$ with root $v$. We call it \emph{neighbour-first search}.

\begin{algorithm}[H] \caption{Neighbour-first search (NFS)} \label{alg:neighbour_first_search}
\textbf{Input:} Connected graph $G$, seed vertex $v$.

\vspace{0.2cm}
\textbf{Output:} Spanning tree $T$ of $G$ with root $v$.

\medskip
NFS($G$, $v$):
\begin{enumerate}
    \item If $T$ is empty, initialize $T = \{v\}$.
    \item For each neighbour $w$ of $v$ which is not yet in $T$, add $w$ to $T$ as a child of $v$.
    \item For each child $w$ of $v$:
    \begin{enumerate}
        \item Do NFS($G$, $w$).
    \end{enumerate}
\end{enumerate}
\end{algorithm}

The spanning tree $T$ output by \Cref{alg:neighbour_first_search} is associated with a partition of the vertex set of $G$ into levels, which are the vertices at a fixed distance from the root of $T$. (The number of levels is the depth of the $T$ plus one.)

\begin{lemma}
Let graph $G=(V,E)$ have $m$ vertices, clique number $\omega$, and longest induced path with $\ell$ vertices.  \Cref{alg:neighbour_first_search} has runtime $O(|V|+|E|)=O(m^2)$ and outputs a spanning tree $T$ of $G$ with the following properties:
\begin{itemize}
    \item The depth of $T$ is no larger than $\ell-1$.
    \item The vertices in any level of $T$ induce a subgraph of $G$ with clique number at most $\omega - 1$.
\end{itemize}
\end{lemma}
\begin{proof}
No vertex can share an edge with any ancestors in $T$ other than its parent, since if it shared an edge with an ancestor higher than its parent then it would have appeared at a higher level as a neighbour of the ancestor. This means a path from the root down $T$ forms an induced path, and the depth of $T$ cannot be longer than the longest induced path in $G$.

Consider two vertices $w_1$ and $w_2$ in the same level $t$ of the spanning tree $T$. Then the children of $w_1$ cannot share any edges in $G$ with the children of $w_2$. This is because the children of $w_1$ constitute a connected component of the subgraph of $G$ induced by all vertices that are not in the first $t$ levels of $T$.
Armed with this observation, consider a clique of size $\omega$ within a level of $T$. When combined with the common parent, this would form a clique of size $\omega+1$ in $G$, a contradiction. Thus the clique number of any level of $T$ is at most $\omega-1$.

Finally, similar to breadth-first search or depth-first search, since each edge is examined at most twice, the time complexity is  $O(|V|+|E|)=O(m^2)$.
\end{proof}

\begin{proof}[Proof of \Cref{lem:chi_binding}.]
We can prove the theorem by induction on $\omega$. Let $\mathcal{A}_\omega$ denote the coloring algorithm which applies to graphs of clique number $\omega$. When $\omega=1$, there are no edges and there is an algorithm $\mathcal{A}_1$ which can color the graph using a single color in $O(m^2)$ time. For the inductive step, assume there is a coloring algorithm $\mathcal{A}_{\omega-1}$ using $\ell^{\omega-2}$ colors and runtime $O(m^2 \omega)$ for any graph of clique number $\omega-1$ and longest induced path $\ell$.

The coloring algorithm $\mathcal{A}_\omega$ for graphs of clique number $\omega$ is as follows. First apply the neighbour-first search algorithm to find spanning tree $T$ of $G$. Then for each level of $T$, apply $\mathcal{A}_{\omega-1}$. For each level, we use a disjoint set of colors. Since there are at most $\ell$ levels in the $T$, the number of colors used by $\mathcal{A}_\omega$ is at most $\ell^{\omega-1}$ by the induction hypothesis.

It remains to analyze the runtime of $\mathcal{A}_\omega$. Say the neighbour-first search step has runtime at most $Cm^2$ in the worst case for some constant $C$. We will show that the runtime of $\mathcal{A}_\omega$ is at most $Cm^2\omega$. From the induction hypothesis, the applications of $\mathcal{A}_{\omega-1}$ have total runtime $\sum_i C m_i^2 (\omega-1) \leq C m^2 (\omega-1)$, where $m_i$ is the number of vertices in the $i^{\text{th}}$ layer. Here we used $\sum_i m_i^2 \leq \big(\sum_i m_i\big)^2 = m^2$. Thus the runtime of $\mathcal{A}_\omega$ is $C m^2 + C m^2 (\omega-1) \leq C m^2 \omega$.
\end{proof}

\section{Learning fermionic obervables} \label{sec:majoranas}

In this section we consider shadow tomography for local fermionic observables.  

A system of $n$ fermionic modes is associated  with a set of $2n$ Majorana fermion operators, which are mutually anticommuting Hermitian observables $\{c_a\}_{a\in [2n]}$ that act on a Hilbert space of dimension $2^n$. We can represent them by a set of Pauli operators satisfying
\begin{equation}
c_1,c_2,\ldots, c_{2n} \in \mathcal{P}^{(n)} \qquad c_ac_b+c_bc_a=2\delta_{ab} \mathbbm{1}.
\label{eq:majoranas2}
\end{equation}
There are a variety of specific choices (fermion-to-qubit mappings) that satisfy the above, including the Jordan-Wigner mapping \cite{jordan1928paulische}, the Bravyi-Kitaev mapping \cite{bravyi2002fermionic}, and the ternary tree mapping \cite{vlasov2019clifford,jiang2020optimal}. The latter two mappings have the desirable property that each Majorana fermion operator is represented by a Pauli operator of low weight $O(\log n)$.

For concreteness, we now review the ternary tree mapping, which is simple enough that it can be understood by looking at \Cref{fig:ternary_tree}. We place a qubit at each non-leaf vertex of a complete rooted ternary tree. The total number of qubits is $n=(3^\ell-1)/2$ where $\ell$ is the depth of the tree. The three edges that connect the vertex for a given qubit $q$ to its children are associated with the three single-qubit Pauli operators $X,Y,Z$ acting on qubit $q$. Each path from the root to a leaf defines an operator which is the tensor product of the Pauli operators which appear on all the edges along the path. One can easily check that these operators are pairwise anticommuting. Moreover, each operator defined in this way has weight exactly $\ell=\log_3(2n+1)$, and there are exactly $3^\ell=2n+1$ of them. We can take all but one of them to be the Majorana fermion operators $\{c_j\}_{j\in [2n]}$.

Below we are interested in $k$-body fermionic observables as defined in \Cref{sec:local}. We write
\begin{equation}
\Gamma(x)= i^{|x|\cdot(|x|-1)/2} c_1^{x_1}c_2^{x_2}\ldots c_{2n}^{x_{2n}} \qquad \qquad x\in \{0,1\}^{2n}
\label{eq:monomials}
\end{equation}
for the Majorana monomials, and 
\begin{equation}
\mathcal{F}^{(n)}_k=\{\Gamma(x): |x|=2k\}
\end{equation}
for the set of $k$-body Majorana operators on $n$ fermionic modes.

\begin{figure}[H]
\centering
  \includegraphics[width=3in]{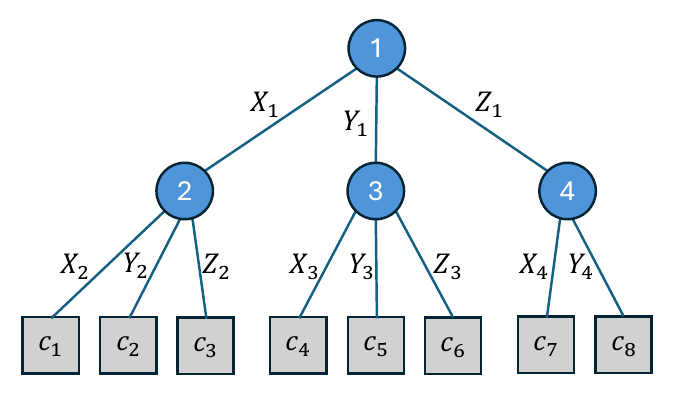}
  \caption{The ternary tree fermion-to-qubit mapping of Refs.~\cite{vlasov2019clifford,jiang2020optimal} for $n=4$ and $\ell=2$. For example, $c_1 = X_1 X_2 = X \otimes X \otimes \mathbbm{1} \otimes \mathbbm{1}$.
  \label{fig:ternary_tree}}
\end{figure}

\subsection{Single-copy lower bounds } \label{sec:lower_bound}

\begin{theorem} \label{thm:local_pauli_lowerbd}
Suppose $1\leq k\leq \log_3(2n+1)$. Any (possibly adaptive) single-copy protocol which learns $\Tr\left(P \rho\right)$ to precision $\eps$ for all $k$-local $n$-qubit Paulis $P \in \mathcal{P}^{(n)}_k$ with constant probability requires $\Omega(3^k / \eps^2)$ copies of $\rho$.
\end{theorem}

\begin{proof}
The ternary tree mapping \cite{vlasov2019clifford,jiang2020optimal} which we reviewed in the previous section shows that, for each $1\leq k\leq  \log_3(2n+1)$ there is a subset $S\subseteq \mathcal{P}^{(n)}_k$ of $3^k$ pairwise anticommuting operators within the set of $k$-local Paulis $\mathcal{P}^{(n)}_k$. See \Cref{fig:ternary_tree}. On the other hand, \Cref{lem:anticommuting_bound} shows that the commutation index of any set of $m$ pairwise anticommuting observables is at most $1/m$. Applying \Cref{thm:CCHL22}, we conclude that  $\Omega(3^k\epsilon^{-2})$ copies of $\rho$ are needed to learn the expected values of all observables from the set $S$. Therefore at least this many samples are needed to learn all expected values of observables in $\mathcal{P}^{(n)}_k$.
\end{proof}

The lower bound in \Cref{thm:local_pauli_lowerbd} matches the sample complexity of known single-copy protocols such as classical shadows, up to a factor of $k\log n$~\cite{huang2020predicting}, see \Cref{table}.

For local fermionic observables, one can obtain a much stronger single-copy sample complexity lower bound which scales polynomially in the system size $n$. This is telling us that local fermionic observables are much harder to learn than local Pauli observables. The following Theorem implies that there is no sample efficient single-copy protocol for $k$-body fermionic observables, as stated in \Cref{thm:notfermions}.

\begin{theorem} \label{thm:local_majorana_lowerbd}
Any (possibly adaptive) single-copy protocol which learns $\Tr\left(\Gamma \rho\right)$ to precision $\eps$ for all $k$-body Majorana operators $\Gamma$ on $n$ fermionic modes with constant probability requires number of copies scaling as $\Omega(n^k / \eps^2)$, for any fixed $k\geq 1$.
\end{theorem}

Establishing \Cref{thm:local_majorana_lowerbd} requires a short detour into \emph{Johnson association schemes}. For a subset $L \subseteq \{0,\dots,q-1\}$, define the \emph{generalized Johnson graph} $G(m,q,L)$ to have vertices corresponding to the subsets of $x \subseteq \{1,\dots,m\}$ of size $|x| = q$, and an edge between any $x$ and $y$ such that $|x \cap y| \notin L$. Notice that the commutation graph of the $k$-body Majorana observables on $n$ fermionic modes is precisely the generalized Johnson graph with $L$ consisting of the \emph{even} integers: $G(\mathcal{F}^{(n)}_k) = G(2n,2k,\{0,2,\dots,2k-2\})$.

In Ref. \cite{linz2024systems}, they show that the Lovasz $\vartheta$-function of the generalized Johnson graph scales like $\vartheta(G(m,q,L)) = \Theta(m^{|L|})$ as $m$ grows for fixed $q$. Seting $L$ equal to the even integers gives the following conclusion for the Lovasz $\vartheta$-function of the commutation graph of degree-$2k$ Majoranas.

\begin{theorem} \label{conjecture}
\emph{(\cite[Theorem 1.2]{linz2024systems})}
For any fixed $k$, $\vartheta(G(\mathcal{F}^{(n)}_k)) = \Theta(n^k)$
as $n$ becomes large.
\end{theorem}

This resolves Conjecture 4.13 of \cite{hastings2022optimizing}, up to the $k$-dependence of the constant factor in $\vartheta(G(\mathcal{F}^{(n)}_k))$.

\begin{proof}[Proof of \Cref{thm:local_majorana_lowerbd}]
\Cref{conjecture} and \Cref{lem:commutation_graph} give
\begin{equation}
\Delta(\mathcal{F}^{(n)}_k) = O(n^k) / \tbinom{2n}{2k} = O(n^{-k}),
\end{equation}
since $|\mathcal{F}^{(n)}_k| = \binom{2n}{2k}$. \Cref{thm:CCHL22} then implies the sample complexity lower bound $\Omega(n^k/\eps^2)$.
\end{proof}

The $\Omega(n^k/\eps^2)$ single-copy sample complexity lower bound in \Cref{thm:local_majorana_lowerbd} matches what is achieved by the single-copy protocols in \cite{bonet2020nearly, jiang2020optimal, wan2022matchgate, huggins2022nearly}, up to a factor of $k\log(n)$, see \Cref{table}. It should also be noted that a matching lower bound of $\Delta(\mathcal{F}^{(n)}_k) = \Omega(n^{-k})$ can be shown by finding a large set of mutually commuting $k$-body fermionic observables.

\subsection{Learning 1-body fermionic observables} \label{sec:1-RDM}

In the case of $1$-body observables there is a simple and practical algorithm for coloring induced subgraphs of $G(\mathcal{F}^{(n)}_1)$ with bounded clique number.

\begin{lemma} \label{lem:quadratic_majorana}
Let $G'$ be any induced subgraph of the commutation graph $G(\mathcal{F}^{(n)}_1)$ of $1$-body fermionic observables, and let $\omega$ be the size of the largest clique in $G'$. There is a classical algorithm with runtime $O(n^2\omega)$ that computes a coloring of $G'$ with at most $\omega + 1$ colors.
\end{lemma}

\begin{proof}

Let $G'=G(S)$ be the subgraph of $G(\mathcal{F}^{(n)}_1)$ induced by some subset $S\subseteq \mathcal{F}^{(n)}_1$ of $1$-body fermionic observables. Let $\omega$ be the maximum size of a clique in $G'$. 

Consider an auxiliary graph $H(S)$ defined as follows. This graph $H(S)$ has $2n$ vertices labeled by the Majorana fermion operators $\{c_1,c_2, \dots, c_{2n}\}$. For each observable $i c_a c_b \in S$, we include an edge $\{c_a, c_b\}$ in $H(S)$. Two elements of $S$ commute if and only if they do not share any Majorana fermion operators. For example, $i c_1 c_2$ anticommutes with $i c_2 c_3$ but commutes with $i c_3 c_4$. Thus a commuting set of $1$-body fermionic observables corresponds to a matching in $H(S)$, and partitioning $S$ into commuting sets corresponds to an edge coloring of $H(S)$.

Now observe that our graph of interest $G'$ is the \emph{line graph} of $H(S)$. An edge coloring of $H(S)$ gives a vertex coloring of $G'$. The edge coloring algorithm of Misra and Gries \cite{misra1992constructive} computes an edge-coloring of a graph $H$ using no more than $\deg(H) + 1$ colors, where $\deg(H)$ is the maximum degree of any vertex in $H$. But the edges connecting to a single vertex in our graph $H(S)$ form a clique in $G'$, so the degree of $H(S)$ is at most $\omega$. The runtime of the edge coloring algorithm is asymptotically upper bounded by the number of vertices times the number of edges, which in our case is $O(n\cdot n\omega)$. 
\end{proof}

It is also possible to directly vertex color the given graph $G'$ using Brooks' theorem, which states that the chromatic number of a graph is at most its maximum degree $+1$, since a high degree vertex also yields a large clique. This argument yields a slightly looser bound of $2\omega$.

\subsection{Learning \texorpdfstring{$k$}{k}-body fermionic observables}
\label{sec:kbody}
We now prove \Cref{lem:chibindfermion_intro}, restated below.

\chibindfermion*

Recall the definition of Majorana monomials from \cref{eq:monomials}. In the following we shall use the commutation relations of these operators which we now derive.
Using \Cref{eq:majoranas} we get
\begin{equation}
c_j \Gamma(y)=(-1)^{|y|+y_j}\Gamma(y)c_j \qquad y\in \{0,1\}^{2n}\quad j\in [2n].
\end{equation}
Applying the above for all indices $j$ in the support of $x\in \{0,1\}^{2n}$ gives
\begin{equation}
\Gamma(x)\Gamma(y)= (-1)^{|x||y|+y\cdot x} \Gamma(y)\Gamma(x) \qquad \qquad x,y\in \{0,1\}^{2n}.
\label{eq:comm}
\end{equation}

\begin{claim}
Suppose $x,y\in \{0,1\}^{2n}$ are such that $|x|, |y|$ are either both even, or both odd. If  $x_j=y_j=0$ for some $j\in [2n]$ then
\begin{equation}
[c_j \Gamma(x), c_j \Gamma(y)]=0 \quad \text{ if and only if } \quad [\Gamma(x), \Gamma(y)]=0.
\end{equation}
\label{claim:1}
\end{claim}
\begin{proof}
Follows directly from \cref{eq:comm}.
\end{proof}

In our proof, it will be helpful to consider Majorana monomials of both odd and even degree. Write
\begin{equation}
    \mathcal{M}^{(n)}_r=\{\Gamma(x): |x|=r, \; x\in \{0,1\}^{2n}\}
\end{equation}
for the set of degree-$r$ Majorana monomials, so that $\mathcal{M}^{(n)}_{2k} = \mathcal{F}^{(n)}_k$ are the $k$-body fermionic observables of interest.

\begin{proof}
The proof is by induction in $r$.  Our inductive hypothesis is that, for any induced subgraph $H$ of the commutation graph $G(\mathcal{M}^{(n)}_r)$ of degree-$r$ Majorana monomials with largest clique of size at most $\omega$, we can sample from a fractional coloring of $H$ with $f_r(\omega)$ colors using a classical algorithm with runtime $t_r(n)$ such that $t_r(n)=\mathrm{poly}(n)$ for any constant $r=O(1)$. Here $f_r(\omega)$ is a polynomial that we determine below. Ultimately we are interested in the even values of $r$ and we have $p_k(\omega)= f_{2k}(\omega)$ where $p_k$ is the polynomial in the statement of \Cref{lem:chibindfermion_intro}.

The base case is $r=2$. We saw in \Cref{lem:quadratic_majorana} that if $S\subseteq \mathcal{M}^{(n)}_2$ and its commutation graph $G(S)$ has no cliques larger than $\omega$, then $G(S)$ can be colored with $\omega + 1$ colors using a classical algorithm with runtime $O(n^2\omega)=O(\mathrm{poly}(n))$ since $\omega\leq |\mathcal{M}_2^{(n)}|=O(n^2)$. Thus $f_2(\omega) = \omega + 1$ and we can efficiently sample from the coloring by selecting a color uniformly at random.

In the following two claims we handle the induction step separately for the odd and even values of $r$.

\begin{claim}\label{claim:odd}
Suppose $r\geq 3$ is odd. Then
\begin{equation}
f_r(\omega)=r\omega f_{r-1}(\omega).
\end{equation}
\end{claim}
\begin{proof}
Let $r\geq 3$ be odd, let $G'$ be an induced subgraph of $G(\mathcal{M}^{(n)}_r)$, and suppose the largest clique in $G'$ has size at most $\omega$. Let $V\subseteq \mathcal{M}^{(n)}_r$ be the vertex set of $G'$. Let $\Gamma(x^1),\Gamma(x^2), \ldots, \Gamma(x^L)\in V$ be a maximal set of pairwise anticommuting operators in $V$. We can construct such a set by starting at any vertex of $G'$ and greedily adding vertices until this is no longer possible. By definition, this set is a clique in $G'$ and therefore $L\leq \omega$. 

Let $I\subseteq [2n]$ be the set of all indices of Majoranas that appear in these operators. Since each has weight $r$, we have
\begin{equation}
|I|\leq r\omega.
\end{equation} 
For convenience let us relabel the Majorana fermion operators so that $I=\{1,2,\ldots, T\}$ where $T\leq r\omega$. Then define
\begin{equation}
S_i=\{\Gamma(z)\in V: z_i=1, \text{ and } z_j=0 \text{ for all } 1\leq j\leq i-1\}.
\label{eq:si}
\end{equation}

We now show that $V$ can be partitioned as
\begin{equation}
V=S_1\sqcup S_2\sqcup \ldots \sqcup S_T.
\label{eq:partition}
\end{equation}
By definition, the sets on the RHS are disjoint and each contained in $V$ so all we need to show is that for any $\Gamma(y)\in V$ there is some $i\in T$ such that $\Gamma(y)\in S_i$. So let $\Gamma(y)\in V$ be given. Since the set $\Gamma(x^1),\Gamma(x^2), \ldots, \Gamma(x^L)\in V$ is a maximal set of pairwise anticommuting operators, we must have 
\begin{equation}
[\Gamma(y), \Gamma(x^j)]=0 \quad \text{ for some $j\in [L]$}.
\label{eq:cxj}
\end{equation}
Since $|y|=|x|=r$ are both odd we see from \cref{eq:comm} that this implies $y\cdot x^j\neq 0$. Therefore $y_i=1$ for some index $i\in \{1,2,\ldots, T\}$. Let $\ell\in [T]$ be the smallest index such that $y_\ell=1$. Then $\Gamma(y)\in S_\ell$ and we have shown $V$ can be partitioned as in \cref{eq:partition}.

Now for each $1\leq i\leq T$ consider the commutation graph $G(S_i)$. Each operator $\Gamma(z)\in S_i$ has $z_i=1$. From \Cref{claim:1}, the commutation graph of $S_i$ is therefore unchanged if we flip $z_i\leftarrow 0$ for all $\Gamma(z)\in S_i$. Define
\begin{equation}
S'_{i}=\{\Gamma(z\oplus \hat{e}_i): \Gamma(z)\in S_i\}.
\end{equation}
We have shown that the commutation graph $G(S_i)$ coincides with the commutation graph $G(S'_i)$, where the set $S_i'\subseteq \mathcal{M}^{(n)}_{r-1}$ only contains degree-$(r-1)$ Majorana monomials. Moreover, $G(V)$ does not contain any clique larger than $\omega$, so neither does its induced subgraph $G(S_i)$. Therefore $G(S'_i)$ does not contain any clique of size greater than $\omega$.

By our inductive hypothesis, for each $1\leq i\leq T$, we can sample efficiently from a fractional coloring of $G(S'_i)=G(S_i)$ with size at most $f_{r-1}(\omega)$. Now let us define a fractional coloring of $V$ in which we choose an index $i\in [T]$ uniformly at random and then sample an independent set in $S_i$ according to the fractional coloring of $G(S_i)$. Note that any independent set in $G(S_i)$ is also an independent set in $G(V)$, so this defines a valid fractional coloring. Moreover, the probability of any vertex $u\in G(V)$ being sampled is equal to the probability that we choose $i$ such that $u\in S_i$ (this probability is $1/T$) times the probability that the sampled independent set of $S_i$ contains $u$ (this is at least $1/f_{r-1}(\omega)$ by our inductive hypothesis). This procedure samples a fractional coloring of size 
\begin{equation}
T\cdot f_{r-1}(\omega)\leq r\omega \cdot f_{r-1}(\omega) 
\end{equation}
as claimed. To sample from the fractional coloring, we need to first construct a maximal set of pairwise anticommuting operators $\Gamma(x^{1}), \ldots, \Gamma(x^L)$, from which we can define the set $I$ and the partition \cref{eq:partition}. As noted above, this step can be performed by starting at an arbitrary vertex $\Gamma(x^1)$ of $G'$ and then growing the set one operator at a time until this is not longer possible. This step has $\mathrm{poly}(n)$ runtime because the graph has at most $|\mathcal{M}^{(n)}_r|=\binom{2n}{r}=\mathrm{poly}(n)$ vertices. The next step is to choose an index $1\leq i\leq T$ at random and sample an independent set of $S_i$ uniformly at random using a fractional coloring of $G(S_i)$ which by our inductive hypothesis can be done in $\mathrm{poly}(n)$ time.
\end{proof}

\begin{claim} \label{claim:even}
Suppose $r\geq 4$ is even. Then 
\begin{equation}
f_r(\omega)=\left(f_{r-1}(\omega)\right)^{r}.
\end{equation}
\end{claim}
\begin{proof}
Let $r\geq 4$ be even, let $G'$ be an induced subgraph of $G(\mathcal{M}^{(n)}_r)$, and let $\omega$ be the size of the largest clique in $G'$. Let $V\subseteq \mathcal{M}^{(n)}_r$ be the vertex set of $G'$. For each $1\leq i\leq 2n$, define 
\begin{equation}
W_i=\{\Gamma(x)\in V: x_i=1\}.
\end{equation}
Note any clique in $G(W_i)$ has size at most $\omega$. From \Cref{claim:1}, the commutation graph of $W_i$ is unchanged if we flip $z_i\leftarrow 0$ for all $\Gamma(z)\in W_i$. Define
\begin{equation}
W'_{i}=\{\Gamma(z\oplus \hat{e}_i): \Gamma(z)\in W_i\}.
\end{equation}

Then $G(W_i)=G(W_i')$ and any clique in $G(W'_i)$ has size at most $\omega$. Moreover, $W_i'$ is a set of degree-$(r-1)$ Majorana monomials, and by our inductive hypothesis we can efficiently sample a coloring of $G(W'_i)$ with size at most $f_{r-1}(\omega)$.  Let $q_i$ be the corresponding fractional coloring of $W_i$, for each $1\leq i\leq 2n$.

Now let us randomly sample a set $\Omega \subseteq V$ as follows. First, select independent sets $I_1\sim q_1, I_2\sim q_2,\ldots ,I_{2n}\sim q_{2n}$ according to the fractional colorings described above. Then let
\begin{equation}
\Omega=\{\Gamma(x)\in V: \Gamma(x)\in I_{j} \text{ for all $j\in [2n]$ such that } x_j=1\}
\end{equation}
Since $|x|=r$ for all $\Gamma(x)\in V$, we have
\begin{equation}
\mathrm{Pr}(\Gamma(x)\in \Omega)\geq  \left(\frac{1}{f_{r-1}(\omega)}\right)^r \qquad \Gamma(x)\in V
\label{eq:colors}
\end{equation}
Now let us show that $\Omega$ is an independent set in $V$; this implies that the above procedure samples from a fractional coloring of V with $(f_{r-1}(\omega))^r$ colors. So suppose $\Gamma(x), \Gamma(y)\in \Omega$. We will show that $\Gamma(x), \Gamma(y)$ commute; equivalently, there is no edge between the corresponding vertices in $G'$. First suppose $x\cap y=\emptyset$. In this case, since $|x|=|y|=r$ are both even, it follows directly that $[\Gamma(x), \Gamma(y)]=0$. If on the other hand $x_j=y_j=1$ for some $j\in [2n]$, Then $\Gamma(x), \Gamma(y) \in W_j\cap \Omega$. But $W_j\cap \Omega \subseteq I_j$ is an independent set in the commutation graph of $W_j$, and therefore $[\Gamma(x),\Gamma(y)]=0$.

The algorithm we have described above only involves identifying the subsets of vertices $W_i$ for $1\leq i\leq 2n$ (which can be done in linear time in the number of vertices of $G'$, which is upper bounded polynomially in $n$), and then using $O(n)$ calls to the subroutine for sampling fractional colorings of commutation graphs of degree-$(r-1)$ Majorana monomials with clique number at most $\omega$. Since this subroutine has $\mathrm{poly}(n)$ runtime by our inductive hypothesis, so does the algorithm described above. 
\end{proof}

Putting together Claims \ref{claim:odd}, \ref{claim:even}, and \Cref{lem:quadratic_majorana} we see that the sizes $f_r(\omega)$ of the fractional colorings are polynomial functions of $\omega$ with degree that depends only on $r$. For even values of $r$ the polynomials $p_k(\omega)=f_{2k}(\omega)$ satisfy the recurrence
\begin{equation}
p_1(\omega) = \omega + 1 \quad \text{ and } \quad p_{k}(\omega) = ((2k - 1) \omega p_{k-1}(\omega))^{2k} \quad \quad k \geq 2.
\end{equation}
For the $2$-body and $3$-body fermionic observables we get
\begin{equation}
p_2(\omega) = O(\omega^8) \quad , \quad p_3(\omega) = O(\omega^{54})
\end{equation}
In general, we have the upper bound
\begin{equation}
p_{k}(\omega) \leq (2k\omega)^{(2k)^{k+1}}.\qedhere
\end{equation}
\end{proof}

\bibliographystyle{alpha}
\bibliography{refs,ryan_ref}

\newcommand{\etalchar}[1]{$^{#1}$}
\begin{thebibliography}{WJHSC16}

\bibitem[Aar04]{aaronson2004limitations}
Scott Aaronson.
\newblock Limitations of quantum advice and one-way communication.
\newblock In {\em Proceedings of 19th IEEE Annual Conference on Computational Complexity (CCC 2004)}, pages 320--332, 2004.

\bibitem[Aar18]{aaronson2018shadow}
Scott Aaronson.
\newblock Shadow tomography of quantum states.
\newblock In {\em Proceedings of the 50th annual ACM SIGACT Symposium on Theory of Computing (STOC 2018)}, pages 325--338, 2018.

\bibitem[ACH{\etalchar{+}}18]{aaronson2018online}
Scott Aaronson, Xinyi Chen, Elad Hazan, Satyen Kale, and Ashwin Nayak.
\newblock Online learning of quantum states.
\newblock In {\em Advances in Neural Information Processing Systems}, volume~31, 2018.

\bibitem[AEHK16]{asadian2016heisenberg}
Ali Asadian, Paul Erker, Marcus Huber, and Claude Kl{\"o}ckl.
\newblock {Heisenberg-Weyl} observables: Bloch vectors in phase space.
\newblock {\em Physical Review A}, 94(1):010301, 2016.

\bibitem[AG04]{aaronson2004improved}
Scott Aaronson and Daniel Gottesman.
\newblock Improved simulation of stabilizer circuits.
\newblock {\em Physical Review A}, 70(5):052328, 2004.

\bibitem[AK07]{arora2007combinatorial}
Sanjeev Arora and Satyen Kale.
\newblock A combinatorial, primal-dual approach to semidefinite programs.
\newblock In {\em Proceedings of the 39th annual ACM Symposium on Theory of Computing (STOC 2007)}, pages 227--236, 2007.

\bibitem[AR19]{aaronson2019gentle}
Scott Aaronson and Guy~N Rothblum.
\newblock Gentle measurement of quantum states and differential privacy.
\newblock In {\em Proceedings of the 51st Annual ACM SIGACT Symposium on Theory of Computing (STOC 2019)}, pages 322--333, 2019.

\bibitem[BGKT19]{bravyi2019approximation}
Sergey Bravyi, David Gosset, Robert K{\"o}nig, and Kristan Temme.
\newblock Approximation algorithms for quantum many-body problems.
\newblock {\em Journal of Mathematical Physics}, 60(3), 2019.

\bibitem[BK02]{bravyi2002fermionic}
Sergey~B Bravyi and Alexei~Yu Kitaev.
\newblock Fermionic quantum computation.
\newblock {\em Annals of Physics}, 298(1):210--226, 2002.

\bibitem[BKL{\etalchar{+}}19]{brandao2017quantum}
Fernando G. S.~L. Brand{\~a}o, Amir Kalev, Tongyang Li, Cedric Yen-Yu Lin, Krysta~M. Svore, and Xiaodi Wu.
\newblock {Quantum SDP Solvers: Large Speed-Ups, Optimality, and Applications to Quantum Learning}.
\newblock In {\em 46th International Colloquium on Automata, Languages, and Programming (ICALP 2019)}, volume 132 of {\em Leibniz International Proceedings in Informatics (LIPIcs)}, pages 27:1--27:14, 2019.

\bibitem[BMBO20]{bonet2020nearly}
Xavier Bonet-Monroig, Ryan Babbush, and Thomas~E O’Brien.
\newblock Nearly optimal measurement scheduling for partial tomography of quantum states.
\newblock {\em Physical Review X}, 10(3):031064, 2020.

\bibitem[BO21]{buadescu2021improved}
Costin B{\u{a}}descu and Ryan O'Donnell.
\newblock Improved quantum data analysis.
\newblock In {\em Proceedings of the 53rd Annual ACM SIGACT Symposium on Theory of Computing (STOC 2021)}, pages 1398--1411, 2021.

\bibitem[BWM{\etalchar{+}}16]{Bauer2016}
Bela Bauer, Dave Wecker, Andrew~J. Millis, Matthew~B. Hastings, and Matthias Troyer.
\newblock Hybrid quantum-classical approach to correlated materials.
\newblock {\em Phys. Rev. X}, 6:031045, Sep 2016.

\bibitem[CCHL22]{chen2022exponential}
Sitan Chen, Jordan Cotler, Hsin-Yuan Huang, and Jerry Li.
\newblock Exponential separations between learning with and without quantum memory.
\newblock In {\em 62nd Annual Symposium on Foundations of Computer Science (FOCS 2022)}, pages 574--585. IEEE, 2022.

\bibitem[CGY24]{Chen2024}
Sitan Chen, Weiyuan Gong, and Qi~Ye.
\newblock Optimal tradeoffs for estimating {P}auli observables, 2024.
\newblock To appear.

\bibitem[CW20]{cotler2020quantum}
Jordan Cotler and Frank Wilczek.
\newblock Quantum overlapping tomography.
\newblock {\em Physical Review Letters}, 124(10):100401, 2020.

\bibitem[dGHG23]{de2023uncertainty}
Carlos de~Gois, Kiara Hansenne, and Otfried G{\"u}hne.
\newblock Uncertainty relations from graph theory.
\newblock {\em Physical Review A}, 107(6):062211, 2023.

\bibitem[DKBC21]{derby2021compact}
Charles Derby, Joel Klassen, Johannes Bausch, and Toby Cubitt.
\newblock Compact fermion to qubit mappings.
\newblock {\em Physical Review B}, 104(3):035118, 2021.

\bibitem[EHF19]{evans2019scalable}
Tim~J. Evans, Robin Harper, and Steven~T. Flammia.
\newblock Scalable {B}ayesian {H}amiltonian learning.
\newblock {\em arXiv preprint arXiv:1912.07636}, 2019.

\bibitem[FB18]{Booth2018}
Edoardo Fertitta and George~H. Booth.
\newblock {Rigorous wave function embedding with dynamical fluctuations}.
\newblock {\em Physical Review B}, 98(23):235132, 12 2018.

\bibitem[GKK{\etalchar{+}}07]{gavinsky2007exponential}
Dmitry Gavinsky, Julia Kempe, Iordanis Kerenidis, Ran Raz, and Ronald De~Wolf.
\newblock Exponential separations for one-way quantum communication complexity, with applications to cryptography.
\newblock In {\em Proceedings of the 39th annual ACM Symposium on Theory of Computing (STOC 2007)}, pages 516--525, 2007.

\bibitem[GS19]{gosset2018compressed}
David Gosset and John Smolin.
\newblock {A Compressed Classical Description of Quantum States}.
\newblock In {\em 14th Conference on the Theory of Quantum Computation, Communication and Cryptography (TQC 2019)}, volume 135 of {\em Leibniz International Proceedings in Informatics (LIPIcs)}, pages 8:1--8:9, 2019.

\bibitem[GWH{\etalchar{+}}16]{ChanChromium}
Sheng Guo, Mark~A. Watson, Weifeng Hu, Qiming Sun, and Garnet Kin-Lic Chan.
\newblock {Electron Valence State Perturbation Theory Based on a Density Matrix Renormalization Group Reference Function, with Applications to the Chromium Dimer and a Trimer Model of Poly-Phenylenevinylene}.
\newblock {\em Journal of Chemical Theory and Computation}, 12(4):1583--1591, 4 2016.

\bibitem[Gy{\'a}87]{gyarfas1987problems}
Andr{\'a}s Gy{\'a}rf{\'a}s.
\newblock Problems from the world surrounding perfect graphs.
\newblock {\em Applicationes Mathematicae}, 19(3-4):413--441, 1987.

\bibitem[HKP20]{huang2020predicting}
Hsin-Yuan Huang, Richard Kueng, and John Preskill.
\newblock Predicting many properties of a quantum system from very few measurements.
\newblock {\em Nature Physics}, 16(10):1050--1057, 2020.

\bibitem[HKP21]{huang2021information}
Hsin-Yuan Huang, Richard Kueng, and John Preskill.
\newblock Information-theoretic bounds on quantum advantage in machine learning.
\newblock {\em Physical Review Letters}, 126(19):190505, 2021.

\bibitem[HO22]{hastings2022optimizing}
Matthew~B. Hastings and Ryan O'Donnell.
\newblock Optimizing strongly interacting fermionic {H}amiltonians.
\newblock In {\em Proceedings of the 54th Annual ACM SIGACT Symposium on Theory of Computing (STOC 2022)}, page 776–789, 2022.

\bibitem[HOR{\etalchar{+}}22]{Huggins2022}
William~J. Huggins, Bryan~A. O’Gorman, Nicholas~C. Rubin, David~R. Reichman, Ryan Babbush, and Joonho Lee.
\newblock {Unbiasing fermionic quantum Monte Carlo with a quantum computer}.
\newblock {\em Nature}, 603(7901):416--420, 3 2022.

\bibitem[HWM{\etalchar{+}}22]{huggins2022nearly}
William~J. Huggins, Kianna Wan, Jarrod McClean, Thomas~E. O’Brien, Nathan Wiebe, and Ryan Babbush.
\newblock Nearly optimal quantum algorithm for estimating multiple expectation values.
\newblock {\em Physical Review Letters}, 129(24):240501, 2022.

\bibitem[JGM19]{jena2019pauli}
Andrew Jena, Scott Genin, and Michele Mosca.
\newblock Pauli partitioning with respect to gate sets.
\newblock {\em arXiv preprint arXiv:1907.07859}, 2019.

\bibitem[JKMN20]{jiang2020optimal}
Zhang Jiang, Amir Kalev, Wojciech Mruczkiewicz, and Hartmut Neven.
\newblock Optimal fermion-to-qubit mapping via ternary trees with applications to reduced quantum states learning.
\newblock {\em Quantum}, 4:276, 2020.

\bibitem[JW28]{jordan1928paulische}
P.~Jordan and E.~Wigner.
\newblock {\"U}ber das {P}aulische {{\"A}}quivalenzverbot.
\newblock {\em Zeitschrift f{\"u}r Physik}, 47(9):631--651, Sep 1928.

\bibitem[KGZ15]{SEET}
Alexei~A. Kananenka, Emanuel Gull, and Dominika Zgid.
\newblock {Systematically improvable multiscale solver for correlated electron systems}.
\newblock {\em Physical Review B}, 91(12):121111, 3 2015.

\bibitem[Kim95]{kim1995ramsey}
Jeong~Han Kim.
\newblock The {R}amsey number {$R(3, t)$} has order of magnitude $t^2/\log t$.
\newblock {\em Random Structures \& Algorithms}, 7(3):173--207, 1995.

\bibitem[Knu93]{knuth1993sandwich}
Donald~E Knuth.
\newblock The sandwich theorem.
\newblock {\em arXiv preprint math/9312214}, 1993.

\bibitem[KSH{\etalchar{+}}06]{DMFTReview}
G.~Kotliar, S.~Y. Savrasov, K.~Haule, V.~S. Oudovenko, O.~Parcollet, and C.~A. Marianetti.
\newblock {Electronic structure calculations with dynamical mean-field theory}.
\newblock {\em Reviews of Modern Physics}, 78(3):865--951, 8 2006.

\bibitem[KWM24]{king2024exponential}
Robbie King, Kianna Wan, and Jarrod McClean.
\newblock Exponential learning advantages with conjugate states and minimal quantum memory.
\newblock {\em arXiv preprint arXiv:2403.03469}, 2024.

\bibitem[Lin24]{linz2024systems}
William Linz.
\newblock {$L$}-systems and the {{L}ov\'asz} number.
\newblock {\em arXiv preprint arXiv:2402.05818}, 2024.

\bibitem[LMC{\etalchar{+}}23]{Voorhis2023}
Yuan Liu, Oinam~R. Meitei, Zachary~E. Chin, Arkopal Dutt, Max Tao, Isaac~L. Chuang, and Troy Van~Voorhis.
\newblock {Bootstrap Embedding on a Quantum Computer}.
\newblock {\em Journal of Chemical Theory and Computation}, 19(8):2230--2247, 4 2023.

\bibitem[MG92]{misra1992constructive}
Jayadev Misra and David Gries.
\newblock A constructive proof of {V}izing's theorem.
\newblock {\em Information Processing Letters}, 41(3):131--133, 1992.

\bibitem[MKCdJ17]{mcclean2017}
Jarrod~R. McClean, Mollie~E. {Kimchi-Schwartz}, Jonathan Carter, and Wibe~A. de~Jong.
\newblock Hybrid quantum-classical hierarchy for mitigation of decoherence and determination of excited states.
\newblock {\em Phys. Rev. A}, 95:042308, Apr 2017.

\bibitem[Mon17]{montanaro2017learning}
Ashley Montanaro.
\newblock Learning stabilizer states by {B}ell sampling.
\newblock {\em arXiv preprint arXiv:1707.04012}, 2017.

\bibitem[PZC23]{ChanRDMTomo}
Linqing Peng, Xing Zhang, and Garnet Kin-Lic Chan.
\newblock {Fermionic Reduced Density Low-Rank Matrix Completion, Noise Filtering, and Measurement Reduction in Quantum Simulations}.
\newblock {\em Journal of Chemical Theory and Computation}, 19(24):9151--9160, 12 2023.

\bibitem[Raz99]{raz1999exponential}
Ran Raz.
\newblock Exponential separation of quantum and classical communication complexity.
\newblock In {\em Proceedings of the 31st annual ACM Symposium on Theory of Computing (STOC 1999)}, pages 358--367, 1999.

\bibitem[SBM06]{SBM06}
V.V. Shende, S.S. Bullock, and I.L. Markov.
\newblock Synthesis of quantum-logic circuits.
\newblock {\em IEEE Transactions on Computer-Aided Design of Integrated Circuits and Systems}, 25(6):1000--1010, 2006.

\bibitem[SKGA17]{sharma2017}
Sandeep Sharma, Gerald Knizia, Sheng Guo, and Ali Alavi.
\newblock Combining internally contracted states and matrix product states to perform multireference perturbation theory.
\newblock {\em Journal of Chemical Theory and Computation}, 13(2):488--498, 2017.

\bibitem[SR19]{schiermeyer2019polynomial}
Ingo Schiermeyer and Bert Randerath.
\newblock Polynomial $\chi$-binding functions and forbidden induced subgraphs: a survey.
\newblock {\em Graphs and Combinatorics}, 35(1):1--31, 2019.

\bibitem[SRL12]{seeley2012bravyi}
Jacob~T. Seeley, Martin~J. Richard, and Peter~J. Love.
\newblock {The Bravyi-Kitaev transformation for quantum computation of electronic structure}.
\newblock {\em The Journal of Chemical Physics}, 137(22):224109, 12 2012.

\bibitem[SS20]{scott2020survey}
Alex Scott and Paul Seymour.
\newblock A survey of $\chi$-boundedness, 2020.

\bibitem[SU11]{scheinerman2011fractional}
E.R. Scheinerman and D.H. Ullman.
\newblock {\em Fractional Graph Theory: A Rational Approach to the Theory of Graphs}.
\newblock Dover books on mathematics. Dover Publications, 2011.

\bibitem[TRJ{\etalchar{+}}20]{Takeshita2020}
Tyler Takeshita, Nicholas~C. Rubin, Zhang Jiang, Eunseok Lee, Ryan Babbush, and Jarrod~R. McClean.
\newblock Increasing the representation accuracy of quantum simulations of chemistry without extra quantum resources.
\newblock {\em Phys. Rev. X}, 10:011004, Jan 2020.

\bibitem[Vla19]{vlasov2019clifford}
Alexander~Yu Vlasov.
\newblock Clifford algebras, spin groups and qubit trees.
\newblock {\em arXiv preprint arXiv:1904.09912}, 2019.

\bibitem[VYI20]{IzmaylovClique}
Vladyslav Verteletskyi, Tzu-Ching Yen, and Artur~F. Izmaylov.
\newblock {Measurement optimization in the variational quantum eigensolver using a minimum clique cover}.
\newblock {\em The Journal of Chemical Physics}, 152(12), 3 2020.

\bibitem[WHLB23]{wan2022matchgate}
Kianna Wan, William~J. Huggins, Joonho Lee, and Ryan Babbush.
\newblock Matchgate shadows for fermionic quantum simulation.
\newblock {\em Communications in Mathematical Physics}, 404(2):629--700, Dec 2023.

\bibitem[WJHSC16]{Wouters2016}
Sebastian Wouters, Carlos~A Jim{\'{e}}nez-Hoyos, Qiming Sun, and Garnet K.-L. Chan.
\newblock {A Practical Guide to Density Matrix Embedding Theory in Quantum Chemistry}.
\newblock {\em Journal of Chemical Theory and Computation}, 12(6):2706--2719, 2016.

\bibitem[XSW23]{xu2023bounding}
Zhen-Peng Xu, Ren{\'e} Schwonnek, and Andreas Winter.
\newblock Bounding the joint numerical range of {P}auli strings by graph parameters.
\newblock {\em arXiv preprint arXiv:2308.00753}, 2023.

\bibitem[YHM{\etalchar{+}}22]{GeneralizedSubspace}
Nobuyuki Yoshioka, Hideaki Hakoshima, Yuichiro Matsuzaki, Yuuki Tokunaga, Yasunari Suzuki, and Suguru Endo.
\newblock {Generalized Quantum Subspace Expansion}.
\newblock {\em Physical Review Letters}, 129(2):020502, 7 2022.

\bibitem[ZRM21]{ZhaoShadow}
Andrew Zhao, Nicholas~C. Rubin, and Akimasa Miyake.
\newblock {Fermionic Partial Tomography via Classical Shadows}.
\newblock {\em Physical Review Letters}, 127(11):110504, 9 2021.

\end{thebibliography}

\appendix

\section{Proof of \texorpdfstring{{\Cref{lem:commutation_graph}}}{{Lemma \ref*{lem:commutation_graph}}}} \label{app:commutation_graph}

\commutationgraph*
\begin{proof}
Denote $S = \{P_1,\dots,P_m\}$. Given $\rho$, let
\begin{equation}
a_j = \Tr\left(P_j \rho\right).
\end{equation}
We aim to show $\sum_j a_j^2 \leq \vartheta(G(S))$.

Consider the observable
\begin{equation}
Q = \sum_j a_j P_j.
\end{equation}
We have
\begin{equation}
Q^2 = \sum_{j,l} a_j a_l P_j P_l = \frac{1}{2} \sum_{j,l} a_j a_l \{P_j,P_l\}.
\end{equation}
Note $P_j^2 = \mathbbm{1}$ since they are Hermitian unitaries.

Now take the trace with $\rho$. We get
\begin{equation}
\Tr\left(Q^2 \rho\right) = \sum_{j,l} a_j a_l B_{jl} \leq \lambda_{\max}(B) \sum_j a_j^2,
\end{equation}
where we defined the matrix
\begin{equation}
B_{jl} = \frac{1}{2} \Tr\left(\{P_j,P_l\} \rho\right).
\end{equation}

By positivity of the state $\rho$, we have $\mathrm{Tr}\big(\big(Q - \Tr\left(Q \rho\right) \mathbbm{1}\big)^2 \rho\big) \geq 0$ and therefore
\begin{align}
\Tr\left(Q \rho\right)^2 &\leq \Tr\left(Q^2 \rho\right) \\
\implies \Big(\sum_j a_j^2\Big)^2 &\leq \lambda_{\max}(B) \sum_j a_j^2 \\
\implies \sum_j a_j^2 &\leq \lambda_{\max}(B).
\end{align}

$B$ satisfies $B_{jj} = 1 \ \forall j$ and $B_{jl} = 0$ for all edges $(j,l)$. The latter holds since $(j,l)$ is an edge precisely when $\{P_j,P_l\} = 0$. Positivity of the state $\rho$ implies that $B$ is positive semidefinite, since for any vector $v \in \mathbb{R}^m$
\begin{align}
v^T B v = \Tr\left(\big(\sum_j v_j P_j\big)^2 \rho\right) \geq 0.
\end{align}
Let's now take the supremum of the right-hand-side over all such $B$ to get
\begin{align}
\sum_j a_j^2 \leq \tilde{\vartheta}(G(S)),
\end{align}
where
\begin{align}
\tilde{\vartheta}(G) = \max \ \{  &\lambda_{\max}(B) , \ B \in \mathbb{R}^{m \times m} \nonumber\\
&\text{s.t.} , \ B_{jj} = 1 \ \forall j , \ B_{jl} = 0 \ \forall (j,l) \in E , \ B \succeq 0  \}.
\label{eq:alternativeL}
\end{align}
\Cref{lem:theta_tilde} completes the proof.
\end{proof}

\begin{lemma} \label{lem:theta_tilde}
\emph{(\cite{knuth1993sandwich})}
The function $\tilde{\vartheta}(G)$ from \Cref{eq:alternativeL} satisfies $\tilde{\vartheta}(G) \leq \vartheta(G)$.
\end{lemma}
\begin{proof}
We will use the dual description \Cref{eq:theta_dual}. Let $(\lambda, A)$ achieve the optimal dual value $\lambda = \vartheta(G)$. Define the $m \times (m+1)$ matrix
\begin{equation}
U = (\vec{1}, \sqrt{\lambda A - {J}}),
\end{equation}
where we padded with the all-ones column vector $\vec{1}$ on the left. (Recall ${J}$ denotes the all-ones matrix.)
Let $B$ be any matrix feasible for $\tilde{\vartheta}(G)$. Decompose
\begin{equation}
B = Q^T D Q = V^T V \quad , \quad V = \sqrt{D} Q,
\end{equation}
where $Q$ is orthogonal and $D$ is diagonal with $D_{11} = \lambda_{\max}(B)$. (The entries of $D$ are the eigenvalues of $B$.)
Now consider the collection of $m$ matrices $\{Y^{(j)}\}$ of size $m \times (m+1)$ given by
\begin{equation}
Y^{(j)}_{ab} = V_{aj} U_{jb}.
\end{equation}
We have
\begin{equation}
\Tr\left((Y^{(j)})^T Y^{(l)}\right) = \Big(\sum_a V_{aj} V_{al}\Big) \Big(\sum_b U_{jb} U_{lb}\Big) = \lambda B_{jl} A_{jl}.
\end{equation}
If $j \neq l$, this is zero, since if $(j,l)$ is an edge in $G$ then $B_{jl} = 0$, and if not then $A_{jl} = 0$. If $j = l$, we get $\Tr\left((Y^{(j)})^T Y^{(j)}\right) = \lambda$. Thus $\{Y^{(j)} / \sqrt{\lambda}\}$ are orthonormal when viewed as vectors of dimension $m(m+1)$, and
\begin{equation}
1 \geq \sum_j \big(Y^{(j)}_{11} / \sqrt{\lambda}\big)^2 = \frac{D_{11}}{\lambda} \sum_j Q_{1j}^2 = \frac{D_{11}}{\lambda} 
\implies \ D_{11} \leq \lambda.\qedhere
\end{equation}
\end{proof}

It is in fact true that $\tilde{\vartheta}(G) = \vartheta(G)$, but we only need $\tilde{\vartheta}(G) \leq \vartheta(G)$ for our purposes.

\section{Learning Greens functions}

A quantity of fundamental interest in the chemistry and physics of fermionic system is the Greens function. The $k$-body Greens function is similar to the $k$-RDM except that some of the operators have been evolved forward in time. Accordingly, the Greens function can be used to characterize the response of a fermionic system to external perturbations.
\begin{definition}
The 1-body Greens function of a state $\rho$ with respect to Hamiltonian $H$ is defined as
\begin{equation}
G_{ab}(t) = \Tr\left(i c_a(t) c_b(0) \rho\right),
\end{equation}
where
\begin{equation}
c_a(t) = e^{iHt} c_a e^{-iHt}.
\end{equation}
\end{definition}

Examples of dynamical properties one can compute from the 1-body Greens function but not the 1-RDM include electrical conductivity, magnetic and electric susceptibility, and dynamic structure factor. The time-dependent part of the Greens function is often essential for characterizing interesting phases of matter and the presence of certain quasiparticles. Many impurity model schemes for converging finite quantum simulations of condensed fermionic systems towards their thermodynamic limit also require the time-dependent part of the Greens function. Such methods include dynamical mean-field theory (DMFT) \cite{DMFTReview}
and self-energy embedding theory \cite{SEET}. Using quantum computers as impurity model solvers in this context has been explored in papers such as \cite{Bauer2016,Voorhis2023}.

In this work we do not give a particularly efficient method for computing Greens functions at non-zero times (in the limit of zero time, the one-body Greens function is the 1-RDM). However, we are able to show that one can compute time-derivatives of the Greens function at $t=0$ for sparse Hamiltonians. One can then use these time derivatives to reconstruct the Greens function using a Taylor expansion. Prior work developing methods for DMFT has used this same approach to reconstructing and embedding Greens functions \cite{Booth2018}.

The value $G_{ab}(0)$ of the Greens function at time zero is simply the 1-RDM, which was tackled in \Cref{lem:quadratic_majorana}.
The $q^{\text{th}}$ derivative at time zero is given by
\begin{equation}
G_{ab}^{(q)}(0) = \Tr\left(i \mathcal{L}_H^q(c_a) c_b \rho\right),
\end{equation}
where $\mathcal{L}_H(X) = i[H,X] = i(HX - XH)$ denotes the \emph{Lie bracket}, or \emph{commutator}. $\mathcal{L}_H^q$ denotes the $t$-fold commutator; for example $\mathcal{L}_H^2(X) = -[H,[H,X]]$.

We will require the Hamiltonian $H$ to be sparse, according to the following definition.

\begin{definition}
A Hamiltonian $H$ is \emph{$s$-sparse} if each Majorana mode appears in at most $s$ terms.
\end{definition}

One should view $\{G_{ab}^{(q)}(0)\}_{a,b}$ as a $n \times n$ matrix for each $q$, and we would like to learn each entry to precision $\eps$. For a given $t$, the naive strategy of measuring one-at-a-time requires $O(n^2 / \eps^2)$ copies of $\rho$. Let Hamiltonian $H$ be $k$-body and $s$-sparse. In this section, we give a quantum algorithm which exploits entangled measurements on $\rho \otimes \rho$ to achieve a sample complexity of $\tilde{O}(\log{n} / \eps^4)$, depending only logarithmically on the system size $n$.

\begin{theorem} \label{thm:greens_fn}
Suppose $\rho$ is an unknown state on $n$ fermion modes, and $H$ a $k$-body and $s$-sparse Hamiltonian. There is an algorithm using entangled measurements on two copies $\rho \otimes \rho$ at a time which can estimate all $\{G_{ab}^{(q)}(0)\}_{a,b}$ to precision $\eps$ with high probability using $\tilde{O}((2skq)^{5q} \log{n} / \eps^4)$ total copies of $\rho$. Moreover, the algorithm runs in time $\poly{n}$.
\end{theorem}

\begin{proof}[Proof of \Cref{thm:greens_fn}]
Let's examine the operators $\{i \mathcal{L}_H^q(c_a) c_b\}_{a,b}$ more closely by expanding
\begin{align}
\mathcal{L}_H^q(c_a) &= \sum_{\Gamma \in S^{(q)}_{H,a}} h^{(q)}_{a,\Gamma} \Gamma \\
i \mathcal{L}_H^q(c_a) c_b &= \sum_{\Gamma \in S^{(q)}_{H,a}} i h^{(q)}_{a,\Gamma} \Gamma c_b \label{eq:greens_expansion}
\end{align}
$S^{(q)}_{H,a}$ denotes the Majorana monomials on which $\mathcal{L}_H^q(c_a)$ has support. Note that all $\Gamma \in S^{(q)}_{H,a}$ have odd degree. Let's assume the original Hamiltonian $H$ was normalized so that the coefficients in the Majorana basis have absolute value at most 1; this implies all $|h^{(q)}_{a,\Gamma}| \leq 1$.

The following lemma makes a crucial observation that the number of terms from $H$ which survive in the expansion of $\mathcal{L}_H^q(c_a)$ is bounded independent of the system size $n$. For example, the terms which survive in $[H,c_a]$ are those which act on the fermion associated to Majorana mode $c_a$, of which there are at most $s$.

\begin{lemma} \label{lem:num_terms}
$|S^{(q)}_{H,a}| \leq s^q (2k)^{q-1} (q-1)!$.
\end{lemma}
\begin{proof}[Proof of \Cref{lem:num_terms}]
The proof is a short combinatorial calculation. The degree of the operators in $S^{(q)}_{H,a}$ are upper bounded by $(2k-2)q + 1$.
This is because we increase the degree by $(2k-2)$ each time we take the Lie bracket, and initially the degree is 1.
Using $s$-sparsity of $H$, we can write a recursion upper bounding $|S^{(q)}_{H,a}|$:
\begin{equation}
|S^{(q)}_{H,a}| \leq |S^{(q-1)}_{H,a}| \cdot s \cdot ((2k-2)(q-1) + 1).
\end{equation}
Using $|S^{(0)}_{H,a}| = 1$, we get
\begin{align}
|S^{(q)}_{H,a}| &\leq s^q ((2k-2)(q-1) + 1) ((2k-2)(q-2) + 1) \dots (2k-1) \\
&\leq s^q (2k)^{q-1} (q-1)!
\end{align}
\end{proof}

The goal is to estimate $\Tr\left(i \mathcal{L}_H^q(c_a) c_b \rho\right)$ to precision $\eps$ for all $(a,b)$. By a triangle inequality on \Cref{eq:greens_expansion}, it is sufficient to estimate $\Tr\left(\Gamma c_b \rho\right)$ to precision $\eps / |S^{(q)}_{H,a}|$ for every $\Gamma \in S^{(q)}_{H,a}$.

Thus we focus on learning the set of Majorana operators
\begin{equation}
\mathcal{S} = \{\Gamma c_b : \Gamma \in S^{(q)}_{H,a} , j = 1,\dots,m\}.
\end{equation}
The number of Majorana operators we are required to learn could be as large as $|S^{(q)}_{H,a}| m$. Thus we do \emph{not} want to measure these one at a time; rather we would like to parallelize the learning of these operators by using entangled measurements. To this end, we establish the following chi-binding result.

\begin{lemma} \label{lem:greens_fn_coloring}
Let $G'$ be any induced subgraph of the commutation graph $G(\mathcal{S})$, and let $\omega$ be the size of the maximal clique in $G'$. Then we can efficiently find a coloring of $G'$ with at most $O(s^q (2k)^{q+2} q^2 (q!) \omega)$ colors.
\end{lemma}
\begin{proof}[Proof of \Cref{lem:greens_fn_coloring}]
We will show that the degree of $G'$ is bounded by $O(s^q (2k)^{q+2} q^2 (q!) \omega)$. Then a greedy coloring algorithm is sufficient to establish the result.

We begin with an initial lemma.

\medskip
\begin{lemma} \label{lem:b_bound}
Fix $\Gamma \in S^{(q)}_{H,a}$. The number of indices $b$ for which $\Gamma c_b$ is a vertex in $G'$ is upper bounded by $2 \omega$.
\end{lemma}

\begin{proof}[Proof of \Cref{lem:b_bound}]
This argument follows the same idea as the degree bound in the proof of \Cref{lem:quadratic_majorana}. Let $B = \{b : \Gamma c_b \in G'\}$. We seek to bound $|B| = O(\omega)$. Let's split $B$ into two subsets:
\begin{align}
B \setminus \Gamma &= \{b : \Gamma c_b \in G', c_b \notin \Gamma\} \\
B \cap \Gamma &= \{b : \Gamma c_b \in G', c_b \in \Gamma\}
\end{align}
(The notation $c_b \in \Gamma$ indicates that $c_b$ appears as a factor in $\Gamma$.) The operators $\{\Gamma c_b\}_{b \in B \setminus \Gamma}$ form a mutually anticommuting set, and likewise for the operators $\{\Gamma c_b\}_{b \in B \cap \Gamma}$. Using the bound on clique size $\omega$, we get $|B \setminus \Gamma| = \omega$ and $|B \cap \Gamma| = \omega$ completing the proof.
\end{proof}
\medskip

Now fix a single operator $\Gamma_0 c_{b_0}$ in $G'$. We will upper bound the degree of $\Gamma_0 c_{b_0}$ in $G'$. Suppose vertex $\Gamma c_b$ forms an edge with $\Gamma_0 c_{b_0}$ in $G'$. This means $\Gamma c_b$ anticommutes with $\Gamma_0 c_{b_0}$. In order for $\Gamma c_b$ and $\Gamma_0 c_{b_0}$ to anticommute, they must overlap on an odd number of Majorana modes.

\medskip
{\bf Case 1.} $\Gamma$ overlaps with $\Gamma_0 c_{b_0}$ on at least one Majorana mode. By a similar combinatorial argument as the proof of \Cref{lem:num_terms}, the number of $\Gamma \in \bigcup_a S^{(q)}_{H,a}$ overlapping with any given Majorana mode is upper bounded by
\begin{equation}
s^q ((2k-2)q + 2) ((2k-2)(q-1) + 2) \dots (2k) \ \leq \ s^q (2k)^q q!.
\end{equation}
Multiplying by the number of single Majorana factors in $\Gamma_0 c_{b_0}$, we get that there are at most
\begin{equation}
s^q (2k)^q q! \cdot ((2k-2)q + 2) \ \leq \ s^q (2k)^{q+1} q (q!)
\end{equation}
$\Gamma \in S^{(q)}_{H,a}$ which overlap with $\Gamma_0 c_{b_0}$ on at least one mode. With $\Gamma$ fixed, there are at most $2 \omega$ choices for $b$ such that $\Gamma c_b \in G'$ by \Cref{lem:b_bound}. Thus there are overall at most
\begin{equation}
2 s^q (2k)^{q+1} q (q!) \omega
\end{equation}
vertices $\Gamma c_b$ forming an edge with $\Gamma_0 c_{b_0}$ such that $\Gamma$ overlaps with $\Gamma_0 c_{b_0}$ on at least one mode.

\medskip
{\bf Case 2.} $\Gamma$ is disjoint from $\Gamma_0 c_{b_0}$. Since $\Gamma c_b$ and $\Gamma_0 c_{b_0}$ anticommute, necessarily $c_b$ must appear as a factor in $\Gamma_0 c_{b_0}$; we can write $c_b \in \Gamma_0 c_{b_0}$.

Fix $b$ such that $c_b \in \Gamma_0 c_{b_0}$, and consider a new graph $G^{(b)}$ whose vertices are the $\Gamma$ such that $\Gamma c_b$ forms an edge with $\Gamma_0 c_{b_0}$ and $\Gamma$ is disjoint from $\Gamma_0 c_{b_0}$. Include an edge between $\Gamma$ and $\Gamma'$ in $G^{(b)}$ if the $\Gamma$ and $\Gamma'$ anticommute. Note that the operators $\Gamma$ and $\Gamma'$ have odd degree as products of Majoranas.

The graph $G^{(b)}$ is very dense, and is close to the complete graph. In fact, the \emph{co-degree} of $G^{(b)}$ is bounded; that is, the degree of the \emph{complement} of $G^{(b)}$. This is because disjoint odd-degree Majoranas anticommute, so in order for $\Gamma$ and $\Gamma'$ \emph{not} to form an edge, the operators $\Gamma$ and $\Gamma'$ must overlap. To examine the co-degree of $G^{(b)}$, we can perform a similar combinatorial calculation as in Case 1 to get an upper bound
\begin{equation}
\codeg(G^{(b)}) \leq s^q (2k)^{q+1} q (q!).
\end{equation}

The co-degree bound means that we can partition $G^{(b)}$ into cliques using at most $\codeg(G^{(b)})+1$ cliques; this corresponds to a greedy coloring of the complement of $G^{(b)}$. To each clique we can apply the assumption on the clique number $\omega$ to get a bound on the overall size of $G^{(b)}$:
\begin{equation}
|G^{(b)}| \leq (\codeg(G^{(b)}) + 1) \cdot \omega \leq O(s^q (2k)^{q+1} q (q!) \omega).
\end{equation}
With at most $(2k-2)q + 2$ choices for $b$, we conclude that there are at most
\begin{equation}
O(s^q (2k)^{q+2} q^2 (q!) \omega)
\end{equation}
vertices $\Gamma c_b$ forming an edge with $\Gamma_0 c_{b_0}$ such that $\Gamma$ is disjoint from $\Gamma_0 c_{b_0}$.
\end{proof}

Combining \Cref{lem:greens_fn_coloring} with \Cref{lem:twocopycolor} completes the proof of \Cref{thm:greens_fn}. At last invoking \Cref{lem:num_terms}, we can bound the final sample complexity as follows:
\begin{align}
N &= \tilde{O}(|S^{(q)}_{H,a}|^4 \log{m} / \eps^4) + O(s^q (2k)^{q+2} q^2 (q!) \cdot |S^{(q)}_{H,a}|^2 / \eps^2) \cdot \tilde{O}(|S^{(q)}_{H,a}|^2 \log{m} / \eps^2) \\
&= \tilde{O}(s^q (2k)^{q+2} q^2 (q!) \cdot |S^{(q)}_{H,a}|^4 \cdot \log{m} / \eps^4) \\
&= \tilde{O}(s^{5q} (2k)^{5q-2} q^3 (q-1)!^5 \log{m} / \eps^4) \\
&= \tilde{O}((2skq)^{5q} \log{m} / \eps^4).
\end{align}
\end{proof}

\end{document}